\documentclass[a4paper]{article}
\usepackage[a4paper,hmargin=3.5cm,vmargin=2.5cm]{geometry}

\usepackage[T1]{fontenc}
\usepackage[utf8]{inputenc}
\usepackage{adjustbox}
\usepackage{algpseudocode}
\usepackage{algorithm}
\usepackage{array}
\usepackage{amsthm}
\usepackage{amsmath}
\usepackage{amssymb}
\usepackage{color}
\usepackage{comment}
\usepackage{chngcntr}
\usepackage{dsfont} %for indicator function
\usepackage[english]{babel}
\usepackage{float}
\usepackage{graphicx}
\usepackage{hhline}
\usepackage{hyperref}
\usepackage{multirow}
\usepackage{url}
\usepackage[caption=false,font=footnotesize]{subfig}
\usepackage[para,online,flushleft]{threeparttable}

\counterwithin{paragraph}{section}
\counterwithin{paragraph}{subsection}

\newcommand{\G}{\mathcal{G}} 
\newcommand{\E}{\mathcal{E}} 
\newcommand{\F}{\mathcal{F}} 
\newcommand{\La}{\mathcal{L}}
\newcommand{\V}{\mathcal{V}} 
\newcommand{\W}{\mathbf{W}} 
\newcommand{\D}{\mathbf{D}}

\newcommand{\Landau}{\mathcal{O}}
\newcommand{\U}{\mathbf{U}}
\newcommand{\Uk}{\mathbf{U}_k}
\newcommand{\R}{\mathbf{R}}
\newcommand{\B}{\mathbf{B}}
\newcommand{\M}{\mathbf{M}}
\newcommand{\lmax}{\lambda_{\rm max}}
\newcommand{\iid}{i.i.d.~}

\newcommand{\Rbb}{\mathbb{R}} 
\newcommand{\Prob}{{\mathbb{P}}}
\newcommand{\Esp}{\mathbb{E}} 

\newcommand{\scp}[2]{\langle #1, #2 \rangle}
\newcommand{\nrm}[2]{\sim \mathcal{N}(#1, #2)}

\DeclareMathOperator*{\spn}{span}
\DeclareMathOperator*{\rnk}{rank}
\DeclareMathOperator*{\vr}{Var}
\DeclareMathOperator*{\cv}{Cov}

\newtheorem{theorem}{Theorem} 
 \newtheorem{lemma}{Lemma}
 \newtheorem{corollary}{Corollary}

\begin{document}

\title{Fast Eigenspace Approximation using Random Signals}

\author{Johan Paratte and Lionel Martin 
\thanks{
EPFL, Ecole Polytechnique Fédérale de Lausanne,
LTS2 Laboratoire de traitement du signal, CH-1015 Lausanne, Switzerland}
\thanks{
The authors contributed equally.}
}

% The paper headers
%\markboth{Journal of \LaTeX\ Class Files,~Vol.~14, No.~8, August~2015}%
%{Shell \MakeLowercase{\textit{et al.}}: Bare Demo of IEEEtran.cls for IEEE Transactions on Magnetics Journals}

% Note that keywords are not normally used for peerreview papers.

\maketitle

\begin{abstract} 
We focus in this work on the estimation of the first $k$ eigenvectors of any graph Laplacian using filtering of Gaussian random signals. We prove that we only need $k$ such signals to be able to exactly recover as many of the smallest eigenvectors, regardless of the number of nodes in the graph. In addition, we address key issues in implementing the theoretical concepts in practice using accurate approximated methods. We also propose fast algorithms both for eigenspace approximation and for the determination of the $k$th smallest eigenvalue $\lambda_k$. The latter proves to be extremely efficient under the assumption of locally uniform distribution of the eigenvalue over the spectrum. Finally, we present experiments which show the validity of our method in practice and compare it to state-of-the-art methods for clustering and visualization both on synthetic small-scale datasets and larger real-world problems of millions of nodes. We show that our method allows a better scaling with the number of nodes than all previous methods while achieving an almost perfect reconstruction of the eigenspace formed by the first $k$ eigenvectors.
\end{abstract}


%%%%%%%%%%%%%%%%%%%%%%%%%%%%%%%%%%%%%%%%%%%
%%      Introduction
%%%%%%%%%%%%%%%%%%%%%%%%%%%%%%%%%%%%%%%%%%%

\section{Introduction}
Although  the questions related to data analytics such as clustering or visualization have received a lot of attention in the past decades, their study is also gaining importance due to the amount of data that one would like to treat nowadays. In particular, this current trend requires that methods must be able to accommodate with large data sets. This imposes two important constraints in the design of new techniques: one must ensure that the complexity and the storage required to process the data are as low as possible.

In the past, many accurate techniques have been introduced to tackle the questions of dimensionality reduction, clustering, and visualization. Mostly, they used the fact shared among those problems that high-dimensional data (in $\Rbb^N$) often admits an accurate low-dimensional intrinsic representation. Finding this embedding alleviate the processing and storage constraints of further processing tasks by representing data points in a space of smaller dimension $d \ll N$.

Eigendecomposition has been at the core of famous techniques used to extract low-dimensional embeddings from high-dimensional data by using the eigenvectors associated with specific eigenvalues. This has been used for partitioning (e.g., spectral clustering~\cite{ng2002spectral, shi2000normalized}), data visualization (e.g., Laplacian eigenmaps~\cite{belkin2001laplacian}), but also simply as a dimensionality reduction technique for preprocessing (e.g., principal components analysis~\cite{jolliffe2002principal}). Alternatively, stochastic algorithms for dimensionality reduction (e.g., stochastic neighbor embedding (SNE) \cite{hinton2002stochastic}) appeared as interesting alternatives, especially for visualization. The main drawback of all the aforementioned techniques is that they tend not to scale well as they have a rather high complexity (e.g., a partial eigendecomposition being $\Landau(kN^2)$ while SNE is $\Landau(N^2)$).

The classical way to recover the eigenspace of a symmetric matrix $\La$ is to diagonalize it as $\La = \U\Lambda \U^*$, with $\U$ being the matrix of eigenvectors and $\Lambda$ the matrix of eigenvalues, and take the first $k$ columns of $\U$. The diagonalization is typically done using a singular value decomposition (SVD) of a symmetric matrix of size $N$ is $\Landau(N^3)$ which is intractable even for medium scale $N$. A great deal of work has been done on faster ways to compute eigenvalues and eigenvectors of $\La$ efficiently (see \cite{bai2000templates} for a review). The fastest methods are variants of Arnoldi or Lanczos iteration methods (\cite{arnoldi1951principle} and \cite{lanczos1950iteration} respectively) such as Implicitly Restarted Arnoldi Method (IRAM) \cite{sorensen1992implicit} or Implicitly Restarted Lanczos Method (IRLM) \cite{calvetti1994implicitly}. The preferred method for graph Laplacians is the IRLM since the matrix is symmetric and sparse most of the time. The IRLM has a worst case complexity of $\Landau(h(|\E|k + k^2N + k^3))$, with $h$ the number of iterations to reach convergence and assuming there are $\Landau(k)$ extra Lanczos steps \cite{bai2000templates}. If we consider sparse graphs with $|\E| \approx \Landau(N)$ and a fixed $k$ independent of the value of $N$, the complexity of the IRLM is bounded by the term $\Landau(k^2N)$.

%\paragraph{Eigenspace approximation}
Since the exact computation of the eigenspace proves to be expensive, several angles were considered to approximate the result. Physicists came with a solution to the problem of eigenspace determination using contour integration techniques for the reduction of the matrices on which to apply the eigendecomposition \cite{peter2014feast}, that allows improving the complexity with almost no loss of precision. Meanwhile, with the additional constraint that the matrix should contain a subset of the columns of the original matrix, Boutsidis et al.~\cite{boutsidis2014near} propose a fast method to approximate low-rank matrix reconstruction (whose optimal solution is the eigenspace generated by the first eigenvectors). Some works, such as \cite{mahoney2010implementing}, focus on their side, on the determination of the first non-trivial eigenvector only. Finally, Bai~\cite{bai2005high} proposes a solution for the approximation of eigenvectors using tridiagonalization of sparse matrices that requires "efficiently" sparse matrices as input. Although this might not necessarily apply in practice depending on the data set at hand, it proved to be efficient in various problems involving modeling physical phenomena with strong locality properties.

% Not eigenspace, subspace that best preserve distances -> trashed for now \cite{feldman2010coresets}

%\paragraph{Distance preservation}
Instead of computing the eigenspace as features of the data points in the new space, distance preservation can be considered sufficient depending on the application. Indeed, for tasks such as clustering, supposing an algorithm such as $k$-means is performed as the final assignment step, the preprocessing for dimensionality reduction only requires pairwise distances between points to be preserved in the new space. In this mind, \cite{ramasamy2015compressive} presents a clustering algorithm that avoids the computation of an SVD by computing polynomial approximations and using the Johnson-Linderstrauss lemma.

On the same line, the authors of \cite{boutsidis2015spectral} show that the power method (computing powers of the normalized weight matrix) gives a good approximation of the eigenvectors for distance preservation. They give a bound on the power required to obtain a good approximation of the clustering. This is among the first works, to our knowledge, to use random signal multiplied by powers of the weight matrix.

Even more recently, \cite{tremblay2016compressive} proposed a fast algorithm for graph clustering which is provably as good as spectral clustering. The first half of their work uses random signal filtering and provides a result similar to the one presented in \cite{boutsidis2015spectral}. Moreover, they additionally show that only a subset of the nodes must be assigned with $k$-means and that the rest can be inferred from the graph structure by solving an optimization problem. They state bounds on the number of signals required and the number of nodes to label with $k$-means.

In this work, we present a new algorithm for Fast Eigenspace Approximation using Random Signals (FEARS) to estimate the first $k$ eigenvectors using random signal filtering techniques that were already used in the works on distance preservation. This time, however, we do not simply find a mapping for distance preservation but we are able to obtain the partial eigenspace, with a total complexity inferior to the previous works.

%\paragraph{Contributions}
In this context, our paper proposes various improvements to the field, whose main contributions are:
\begin{itemize}
\item a very efficient scheme for the estimation of eigenspaces using filtering of random graph signals
\item a proven tight bound for the number of random signals needed for perfect recovery
\item algorithms and implementations with practical considerations regarding filter design, fast filtering, and numerical stability  
\item an accelerated method for the count of eigenvalues in a given range
%\item an illustration of our proposed method validity on visualization and clustering tasks
% \item a new quantitative measure of the "smoothness" of the visualization
%\item experiments on synthetic and real data sets showing the superior scalability of this method compared to the state of the art
\end{itemize}

%\paragraph{Organization}
The paper is organized as follows. In Section~\ref{sec:background}, we recall the fundamentals of graph signal processing and define the notation. Section~\ref{sec:contrib} develops the main results of this paper from the theoretical point of view while Section~\ref{sec:practical} presents its applied counterpart and also presents the algorithms for fast spectral embedding and eigencount estimation. Later in Section~\ref{sec:experiments}, we show the validity and benefits of our method and compare with the state of the art through several experiments. Finally, Section~\ref{sec:conclusion} proposes interesting open problems in the domain as well as potential future work to address.

\section{Background} \label{sec:background}

\paragraph{Graph nomenclature}
Let us define $\G = (\V, \E, \W)$ as an undirected weighted graph where $\V$ is the set of vertices and $\E$ the set of edges representing connections between nodes in $\V$. The vertices $v \in \V$ of the graph are ordered from $1$ to $N=|\V|$. The matrix $\W$, which is symmetric and positive, is called the weighted adjacency matrix of the graph $\G$. The weight $\W_{ij}$ represents the weight of the edge between vertices $v_i$ and $v_j$ and a value of 0 means that the two vertices are not connected. The degree $d(i)$ of a node $v_i$ is defined as the sum of the weights of all its edges $d(i)=\sum_{j=1}^N \W_{ij}$. Finally, a graph signal is defined as a vector of scalar values over the set of vertices $\V$ where the $i$-th component of the vector is the value of the signal at vertex $v_i$.

\paragraph{Spectral theory}
The combinatorial Laplacian operator $\La$ can be defined from the weighted adjacency matrix as $\La = \mathbf{D}-\W$ with $\mathbf{D}$ being the degree matrix defined as a diagonal matrix with $\D_{ii}=d(i)$. One alternative and often used Laplacian definition is that of the normalized Laplacian $\La_n = \mathbf{D}^{-\frac{1}{2}} \La \mathbf{D}^{-\frac{1}{2}} = \mathbf{I} - \mathbf{D}^{-\frac{1}{2}} \W \mathbf{D}^{\frac{1}{2}}$. Since the weight matrix $\W$ is symmetric positive semi-definite, so is $\La$ by construction. By application of the spectral theorem, we know that $\La$ can be decomposed into an orthonormal basis of eigenvectors noted $\{ \mathbf{u}_\ell \}_{\ell=0, 1,\ldots, N-1}$. The ordering of the eigenvectors is given by the eigenvalues noted $\{ \lambda_\ell \}_{\ell=0, 1,\dots, N-1}$ sorted in ascending order $0=\lambda_0 \leq \lambda_1 \leq \lambda_2 \leq \ldots \leq \lambda_{N-1} = \lambda_{\rm max}$. In a matrix form we can write this decomposition as $\La = \U\Lambda \U^*$ with $\U = (\mathbf{u}_1 | \mathbf{u}_2 | \ldots | \mathbf{u}_{N-1} )$ the matrix of eigenvectors and $\Lambda$ the diagonal matrix containing the eigenvalues in ascending order. Given a graph signal $f$, its graph Fourier transform is thus defined as $\hat{f} = \F (f) = \U^*f$, and the inverse transform $f = \F^{-1}(\hat{f})=\U\hat{f}$. It is called a Fourier transform by analogy to the continuous Laplacian whose spectral componants are Fourier modes, and the matrix $\U$ is sometimes referred to as the graph Fourier matrix (see e.g., \cite{chung1997spectral}). By the same analogy, the set $\{ \sqrt{\lambda_\ell} \}_{\ell=0, 1,\ldots, N-1}$ is often seen as the set of graph frequencies~\cite{shuman2013vertex}.

\paragraph{Graph filtering}
In traditional signal processing, filtering can be carried out by a pointwise multiplication in Fourier. Thus, since the graph Fourier transform is defined, it is natural to consider a filtering operation on the graph using a multiplication in the graph Fourier domain. To this end, we define a graph filter as a continuous fonction $g:\Rbb_+ \rightarrow \Rbb$ directly in the graph Fourier domain. If we consider the filtering of a signal $f$, whose graph Fourier transform is written $\hat{f}$, by a filter $g$ the operation in the spectral domain is a simple multiplication $\hat{f'}[\ell] =  g(\lambda_\ell) \cdot \hat{f}[\ell]$, with $f'$ and $\hat{f'}$ the filtered signal and its graph Fourier transform respectively. Using the graph Fourier matrix to recover the vertex-based signals we get the explicit matrix formulation for graph filtering:
$$f' = \U g(\Lambda) \U^* f ,$$
where $g(\Lambda) = \text{diag}(g(\lambda_0), g(\lambda_1), \ldots, g(\lambda_{N-1}))$. The graph filtering operator $g(\La) := \U g(\Lambda) \U^*$ is often used to reformulate the graph filtering equation as a simple vector-matrix operation $f' = g(\La) f$.

Since the filtering equation defined above involves the full set of eigenvectors $\U$, it implies the diagonalization of the Laplacian $\La$ which is costly for large graphs. To circumvent this problem, one can represent the filter $g$ as a polynomial approximation, since polynomial filtering only involves the multiplication of the signal by a power of $\La$ of the same order as the polynomial. Filtering using good polynomial approximations can be done using Chebyshev or Lanczos polynomials \cite{hammond2011wavelets, susnjara2015accelerated}.
% These methods scale with the number of edges $|E|$ and reduce the complexity to $\Landau(|E|)$, which is advantageous in the case of sparse graphs.

\section{Eigenspace estimation using random signals} \label{sec:contrib}

The goal of our method is to get the best estimation of the subspace of the graph Laplacian $\La$, denoted $\Uk$, for the lowest computational cost. In a similar approach to \cite{tremblay2016compressive} and \cite{boutsidis2015spectral}, we consider the filtering of random signals. We chose an ideal low-pass filter $g(\La) = \Uk \Uk^*$ to achieve this goal. Throughout this section, we prove the following theorem, one of our main results:
\begin{theorem} \label{thm:main}
Let $g$ be an ideal low-pass filter of cutoff frequency $\lambda_k$, let $\R \in \Rbb^{N\times d}$ a random matrix formed of entry-wise independent and identically distributed Gaussian random variables $\nrm{0}{\frac{1}{d}}$. Let $\La$ be the Laplacian of any graph $\G$.

For any $d \geq k$, performing a QR decomposition on the result of the filtering of $\R$ by $g$ provides the first $k$ eigenvectors of $\La$ altered only by a rotation in $\Rbb^k$.
\end{theorem}

\subsection{Exact eigenspace recovery with random signals} \label{sec:theory}

Assuming we pack $d$ Gaussian random signals with \iid entries $\nrm{0}{\frac{1}{d}}$ in a Gaussian random matrix $\R \in \Rbb^{N \times d}$, the result of the filtering using the filter $g$ can be written as $\M = \Uk \Uk^* \R = \Uk \R_k$. We will first state a result regarding $\R_k$ and then use $\R_k$ directly to compute the projection.

\begin{lemma} \label{lm:rk_ortho}
Let $\U$ be an orthonormal basis and denote $\Uk$ a subset of $k$ of its rows.

The projection of a Gaussian random matrix $\R\nrm{0}{\sigma^2 I}$ onto $\Uk$ preserves all the Gaussian properties of $\R$.
\end{lemma}
\begin{proof}
The multiplication of a Gaussian random matrix by a basis such as $\U$ preserves all the properties of the initial random matrix (Gaussian, entry-wise independence, identical mean,  variance, and size). This proof can be found in the appendix of this paper.

Selecting any subset of the rows of $\U$ changes the size but conserves the orthonormal properties over the rows. Indeed, without loss of generality on the rows selection, we have

\begin{equation}
\begin{pmatrix} I_k\\ 0 \end{pmatrix} \U \R = \begin{pmatrix} I_k\\ 0 \end{pmatrix} \R' = \R_k
\end{equation}

Thus, only the size will be altered compared to a multiplication by the full matrix $\U$. This concludes the proof.
\end{proof}

With lemma \ref{lm:rk_ortho}, we have that $\R_k \in \Rbb^{k\times d}$ is \iid Gaussian of zero mean and variance $\frac{1}{d}$. The next step is to show that $\R_k$ is full rank.

\begin{lemma} \label{lm:rank_rk}
    Let $\R_k \in \Rbb^{k\times d}, d \geq k$ be a Gaussian random matrix of entry-wise \iid $\nrm{0}{\sigma^2}$.
    
    $R_k$ is a full rank matrix with probability 1. That is, $\rnk\{\R_k\} = k$ since $d \geq k$.
\end{lemma}
\begin{proof}
Let us consider the limit case $d = k$. In this case we have to show that the square ($k\times k$) matrix $\R_k$ is non-singular. Indeed, the set of singular Gaussian random matrices $\mathcal{R}_s = \{ \R_k : \det( \R_k ) = 0\}$ is of dimension $k-1$ since it is generated by the zeros of a polynomial of order $k$. Moreover, since the complete set $\mathcal{R} = \{ \R_k \}$ has dimension $k$, the codimension of $\mathcal{R}_s$ is $1$. Thus, the set $\mathcal{R}_s$ is a null set, which means that picking a matrix at random from the set $\mathcal{R}$ returns a matrix from $\mathcal{R}_s$  with probability 0. Hence, $\R_k$ is non-singular with probability 1.

If we consider $d > k$, any square matrix formed of $k$ of the columns of $\R_k$ has rank $k$ following the proof above for the square case. Now, adding columns to this matrix can not change the rank since it can not reduce it and the matrix is full rank already.
\end{proof}

%\subsection{Numerical limits of rank approximation}
As lemma \ref{lm:rank_rk} is critical to the proof, we make a slight digression and discuss its numerical approximation. Indeed, we proved that the matrix $\R_k$ is full rank and this means that the smallest singular value of $\R_k$ is strictly positive. However, while computing singular value decomposition, numerical approximations are performed and the singular values below a given threshold are assimilated to linearly dependent columns. In other words, we need to make a stronger statement and ensure that the smallest singular value stays above a numerical precision threshold in good probability.

To this end, we recall the result of \cite[lemma 3.15]{martinsson2011randomized}:

\begin{lemma}
    Suppose that $k$ and $\ell$ are positive integers with $k \leq \ell$. Suppose further that $G$ is a real $\ell\times k$ matrix whose entries are \iid Gaussian random variables of zero mean and unit variance, and $\beta$ is a positive real number, such that
        \begin{equation}
                \label{eq:martinsson}
                1 - \frac{1}{\sqrt{2\pi(\ell - k +1)}} \Big( \frac{e}{(\ell - k + 1) \beta} \Big)^{\ell - k + 1}
        \end{equation}
is nonnegative.

Then, the least (that is, the $k$th greatest) singular value of G is at least $\frac{1}{\sqrt{\ell}\beta}$ with probability not less than the amount in \eqref{eq:martinsson}.
\end{lemma}

Since $\R_k$ in our case is a Gaussian random matrix of size $k \times k$, zero mean and variance $\frac{1}{k}$, then $s_{\min}$, the smallest singular value of $\R_k$, equals $\frac{\lambda_{\min}}{\sqrt{k}}$ with $\lambda_{\min}$ the smallest singular value of a matrix whose entries match the lemma above. Thus from this result, we can state that the cummulative density function of $s_{\min}$ is:
\begin{equation}
    \Prob(s_{\min} < \frac{1}{\beta}) < \frac{e}{\beta\sqrt{2\pi}}
\end{equation}
In practice, we need to ensure that the minimal singular value is above the predefined threshold of the rank estimate, that usually is around $10^{-13}$. Knowing that the probability of $s_{\min}$ being below the numerical threshold is less than $\frac{e}{\sqrt{2\pi}} 10^{-13} \approx 10^{-13}$, we can conclude that the claim we made theoretically for the rank still holds in practice with very high probability.

Now that we confirmed that $\R_k$ is full rank, even considering numerical approximations, we analyze the final projection $\M = \Uk\R_k$. 

\begin{lemma}
    Let $\M = \Uk\R_k$ a matrix of size $N\times d$, with $\Uk$ and $\R_k$ as defined above. The two following statements are correct:
\begin{align}
    \forall x \in \Rbb^k, \exists y \in \Rbb^d: \Uk x = \M y.\\
    \forall y \in \Rbb^d, \exists x \in \Rbb^k: \Uk x = \M y.
\end{align}
    That is $\M$ and $\Uk$ share the same column space.
\end{lemma}
\begin{proof}
Since $\R_k$ is full rank, its span is able to generate any matrix of $\Rbb^{k\times d}$. Then, the projection of this full space onto $\Uk$ can form any matrix generated by the span of $\Uk$.
\end{proof}

Note that, although all the lemmas above assume $d \geq k$, we suggest using $d = k$ in practice since this is the minimal value for which the result holds and thus the one that will require the least computation.

\begin{proof}[Proof of theorem \ref{thm:main}]
From $\M$, we can find a set of $k$ orthonormal vectors $\B = \{ \mathbf{b_1} | \mathbf{b_2} | \ldots |\mathbf{b_k} \}$, e.g., by applying an SVD. We obtain a decomposition such as $\M = \B \Sigma \mathbf{V}^\top$, with $\Sigma$ a diagonal matrix and $ \mathbf{V} $ an orthogonal matrix. This gives the following equality:
\begin{equation}
    \Uk \R_k = \M = \B \Sigma \mathbf{V}^\top
\end{equation}
and thus $\Uk$ and $\B$ have the same column space by definition. But since $\B$ and $\Uk$ also have the same shape and orthonormal columns, they necessarily relate to each other as $\B = \Uk \mathbf{Q}$, for some rotation matrix $\mathbf{Q} \in \Rbb^{k\times k}$.
\end{proof}

Before moving on to the following of the paper, we would like to stress the fact that the theory described here does not use any assumption made on $\La$. Thus, the statements we make are also true for any matrix for which there exists a spectral decomposition. However, the sparsity of this matrix is key to a fast implementation using graph filtering as we will show next.

\subsection{$\M$ as an approximation of $\Uk$}
\label{sec:mapprox}
The matrix $\B$ has been shown to approximate $\Uk$ up to a rotation, which is perfectly fine for all common applications (e.g., embedding, spectral clustering, etc.). In the following lines, we wanted to present the quality of $\M$ as a direct approximation of $\Uk$. In the discussion below, we show that it could be enough in some situations to stop the procedure before the SVD step and reduce then the complexity of the algorithm.

Recall that $\M$ and $\B$ share the same column space (i.e., $\spn\{\Uk\}$) as we proved in Theorem~\ref{thm:main} and have the same shape. The major difference between the two is that only the latter is composed of normalized columns. However, the distribution of the singular values of $\M$ is well known: it is the same as that of $\R_k$ since $\Uk$ has unitary columns. Moreover, the works of Marchenko and Pastur \cite{marchenko1967distribution} contain lots of results regarding the study of Gaussian ensemble and Wishart matrices. They showed, among other things, that the eigenvalues of Wishart matrices follow a quarter circle law, which means that the distribution of any singular value of $\M$ is a normalized quarter circle of support $[0; 2]$ when $d=k$. On top of that, they proved that the expected value and the standard deviation of those eigenvalues tend to 1 as $N$ becomes large. This means that in average, even with $d=k$, $\M$ is a very good candidate for the approximation of the subspace. The problem is that with the variance on the eigenvalue distribution, random samples hardly benefit from the expectation.

Meanwhile, the Johnson-Linderstrauss lemma says that with $d=\Landau(\log(N))$, the distances are almost preserved (up to a $(1+\varepsilon)$ multiplicative factor) with high probability between rows of $\Uk$ and rows of $\M$. Thus, it seems intuitive that picking more random signals would improve the repartition of the eigenvalues between 0 and 2 and concentrate around the mean. In fact, from the definition of the Marchenko-Pastur distribution, we have the following result:

\begin{corollary}[Corollary 5.35 from \cite{vershynin2010introduction}] \label{cor:eigenconcentration} 
Let $A$ be an $N\times n$ matrix whose entries are independent standard normal random variables. Then for every $t \geq 0$, with probability at least $1 -2 \exp(-t^2/2)$ one has
$\sqrt{N} - \sqrt{n} - t \leq s_{\min}(A) \leq s_{\max}(A) \leq \sqrt{N} + \sqrt{n} + t$.
\end{corollary}

In our case, the entries are Gaussians of variance $\frac{1}{N}$ and the result becomes:
\begin{equation}
1 - \sqrt{\frac{k}{d}} - \frac{t}{\sqrt{d}} \leq s_{\min}(\R_k) \leq s_{\max}(\R_k) \leq 1 + \sqrt{\frac{k}{d}} + \frac{t}{\sqrt{d}}
\end{equation}

We conclude that the more the matrix $\R_k$ is flat (i.e., $d > k$), the more its eigenvalues are concentrated around 1 in good probability, which confirms our intuition.  %Thus, we can compute a bound on the distance preservation with this result since $(1 - \sqrt{\frac{k}{d}} - \frac{t}{\sqrt{d}}) \|\Uk^\top x\|_2 \leq \|\M^\top x\|_2 \leq (1 + \sqrt{\frac{k}{d}} + \frac{t}{\sqrt{d}}) \|\Uk^\top x\|_2$.

\section{Computational aspects of subspace approximation}
\label{sec:practical}

Now that our main theoretical result is established, we look into its practical implementation, while focusing on efficient solutions. First, we present a solution on how to find the cutoff eigenvalue $\lambda_k$. Then, we show our choice for the actual filter design using polynomials enabling fast filtering operations while limiting the problems caused by the approximation. Finally, we describe our algorithms and analyze their complexity. 

\subsection{Estimation of $\lambda_k$} \label{sec:lambdak}

The computation of $\M$ described above depends on the quality of the filter $g$ and the determination of its cutoff frequency $\lambda_k$ which is not known a priori. A standard method is to use eigencount techniques such as the one proposed in \cite{di2016efficient}. In this work, the authors used the fact that the energy retained by an ideal low-pass filtering of random signals with cutoff frequency $\lambda_k$, called $g_{\lambda_k}$, is proportional to the number of eigenvalues that are smaller than $\lambda_k$. Mathematically, we have:
\begin{equation}
    \Esp[\|g_{\lambda_k}(\La)\R\|_F^2] = |\{\lambda : \lambda \leq \lambda_k\}|.
\end{equation}
Thus, by dichotomy, one could approximate the desired threshold value $\lambda_k$ for our filter since we want it to capture exactly $k$ eigenvalues and we know that $\lmax \leq 2$ for normalized Laplacians. Unfortunately, each step of the dichotomy requires $\Landau(k)$ filterings and the dichotomy must be applied $\Landau(\log(N))$ times, without making strong assumptions on the distribution of the eigenvalues over the spectrum. Thus, the estimation of $\lambda_k$ such as defined and used in \cite{tremblay2016compressive} is $\Landau(m|\E|k\log(N))$, which is above the complexity of all the rest of our problem.

We propose now an accelerated version of the eigencount technique for the determination of the threshold of the filter that will not increase the complexity of the overall algorithm. We first assume that the eigenvalues are distributed evenly over the spectrum (between 0 and $\lmax$). Thus, on average, the $k^\text{th}$ eigenvalue should be $\Esp(\lambda_k) = \frac{k}{N}\lmax$. However, one will not find the exact count systematically on the first guess, due to the randomness of the process and the non-uniformity of the eigenvalue distribution in practice. We suggest thus to iterate with the assumption of local uniformity of the distribution of the eigenvalues until the goal is reached. In practice, this means that after picking $\lambda(0) = \frac{k}{N}\lmax$, one should apply the eigencount technique to compute the approximation of the real number of eigenvalues below $\lambda(t)$ in the graph of study, called $n_\lambda(t)$, and iterate with $\lambda(t+1) = \frac{k}{n_\lambda(t)} \lambda(t)$ until the targeted count is achieved with good precision (see Algorithm~\ref{algo:est_lk} for details). As the number of iterations does not depend on $N$ but only of the local eigenvalue distribution, a good precision can be achieved with a constant number of iterations. The cost in number of operation of this accelerated version is thus $\Landau(m|\E|k)$ which is acceptable since it is of the same order than the remaining of our method.

\subsection{Acceleration using fast filtering} \label{sec:filtering}

The construction of the matrix $\M$ in the previous section requires the knowledge of the first $k$ eigenvectors of the graph Laplacian. This knowledge is very costly for large graphs ($N$ large) since it requires a partial SVD of a $N\times N$ matrix, which we try to avoid in the first place. Fortunately, as we explained before, the product $\Uk \Uk^\top$  corresponds to a graph filtering with $g(\La)$, $g$ being an ideal low-pass filter:
$$g(\lambda) = \begin{cases}
    1 & \lambda \leq \lambda_k\\
    0 & \lambda > \lambda_k
\end{cases}$$

Since we cannot afford the cost of exact filtering, we use a polynomial approximation of the filter $g(\La)$. There exist several methods using powers of the Laplacian that allow approximating such filters with polynomials (Chebyshev~\cite{hammond2011wavelets}, Jackson-Chebyshev~\cite{di2016efficient} or Lanczos~\cite{susnjara2015accelerated} polynomials). In the task at hand, the Jackson-Chebyshev polynomial approximation is the best suited to approximate the step function of $g(\La)$ since it avoids the Gibbs effect of Chebyshev polynomials as can be seen in Fig.~\ref{fig:poly_approx_gibbs}.

\begin{figure}[t!]
\begin{center}
\subfloat[]{
    \includegraphics[width=0.6\textwidth]{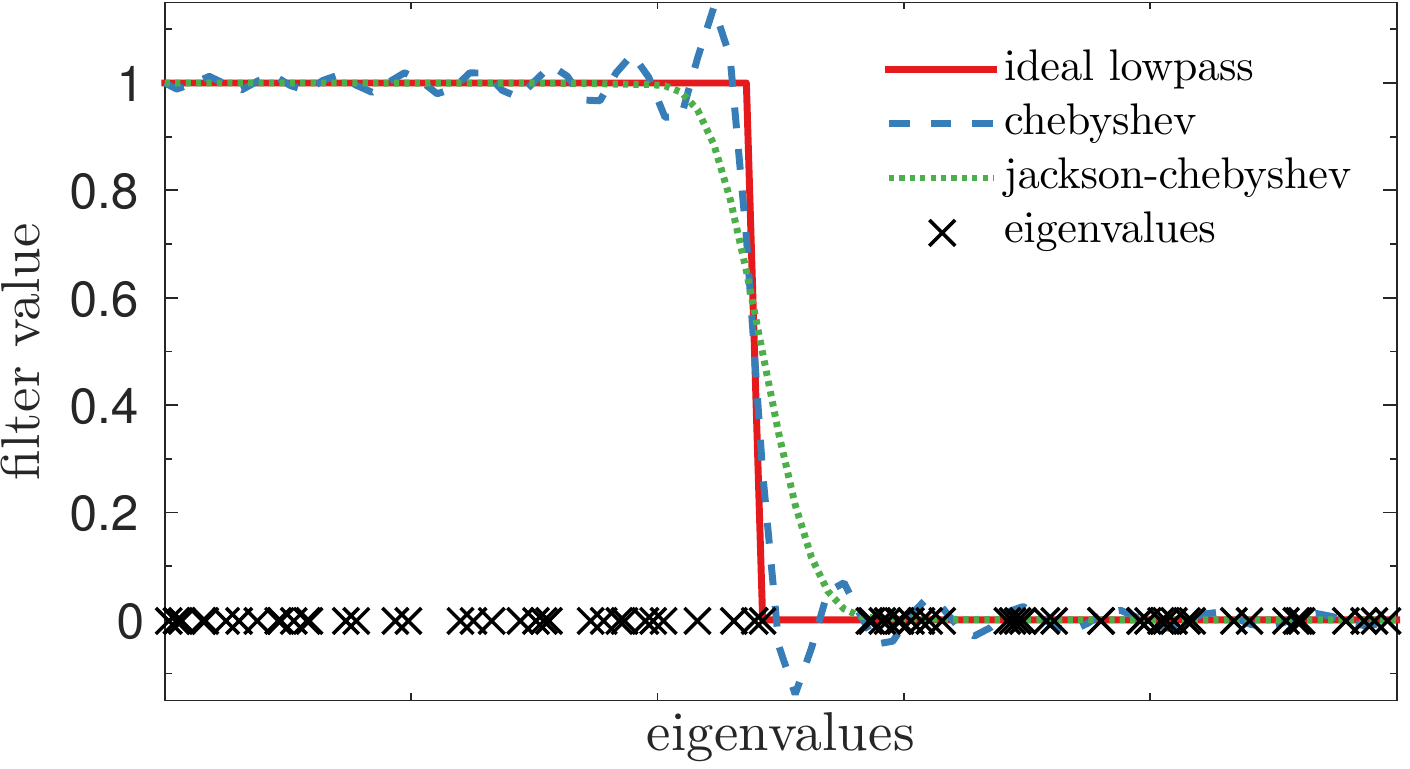}
    \label{fig:poly_approx_gibbs}
}
\hfil
\subfloat[]{
    \includegraphics[width=0.6\textwidth]{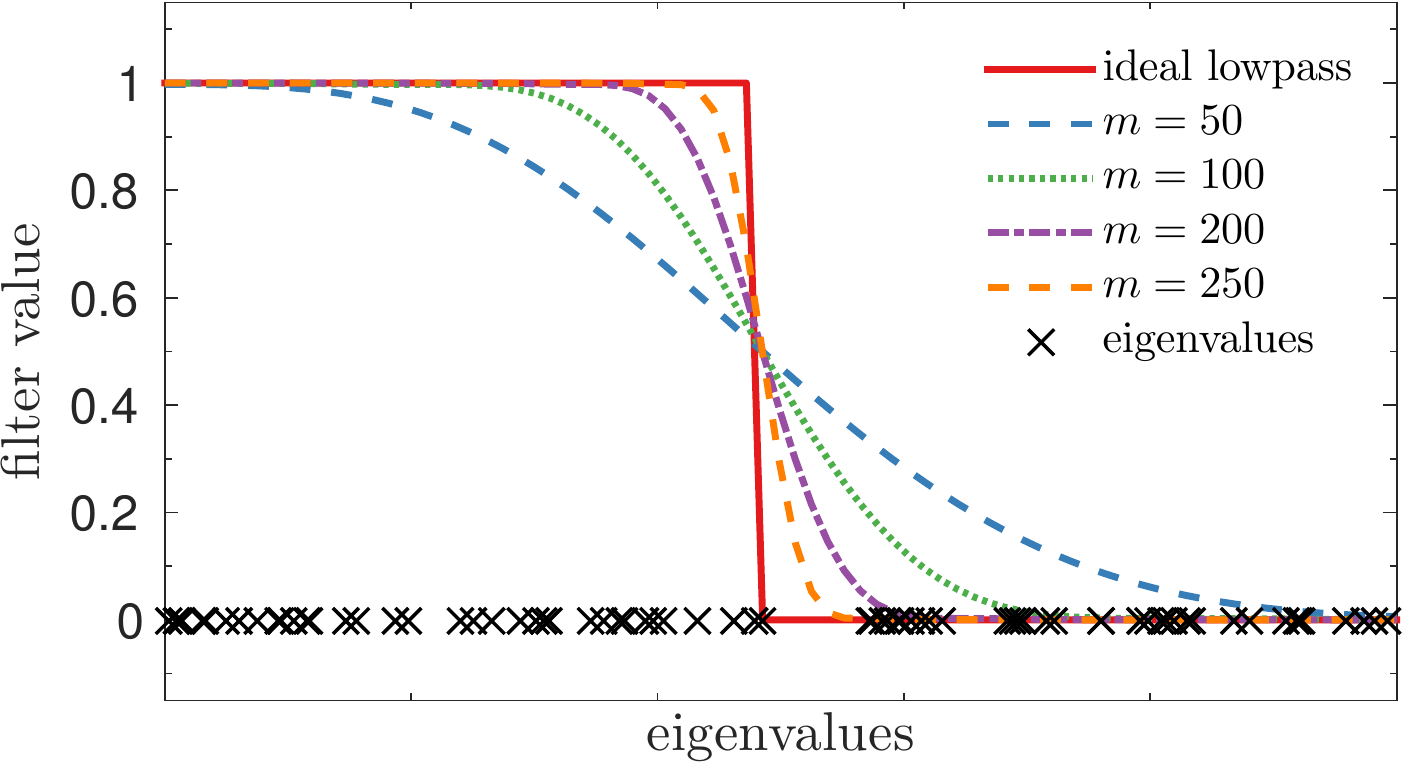}
    \label{fig:poly_approx_order}
}
\caption{The effect of approximating a step function with polynomials. The solid red line is the ideal step function. The black crosses represent the eigenvalues. The approximation using Jackson-Chebyshev polynomials (dotted line) is compared with Chebyshev polynomial approximation (dashed line) of same order $m$ in (a). Jackson-Chebyshev approximations with different orders $m$ are compared in (b). }
\label{fig:poly_approx}
\end{center}
\end{figure}

The quality of the approximation is based on the order of the polynomial, directly related to the number of coefficients to compute. If we define $m$ as the highest degree of the polynomial, we can show that the error of approximation decreases as $m$ increases. This effect is shown in Fig.~\ref{fig:poly_approx_order} where we can see the convergence to the ideal low-pass with an increasing value of $m$. But since the complexity of the filtering increases linearly with $m$ one cannot let it become too large. In particular, we cannot let $m$ be $\Landau(N)$ since it would have a huge impact on the overall complexity.

Let us remind here that the filter approximation needs to be correct only on the discrete values given by the eigenvalues. Indeed, the approximation does not need to fit closely $g$ while the discrete values that the filter takes on the eigenvalues are correct. In our case, since we only want to approximate a step function, we need the value of the filter to be equal to $1$ for $\lambda_0, \lambda_1, \ldots, \lambda_{k-1}$ and $0$ for $\lambda_k, \lambda_{k+1}, \ldots, \lambda_{N-1}$. Two situations could lead to the non-respect of this condition. The estimated cutoff eigenvalue can be wrong or the order of the polynomial can be too small. 

If the order $m$ is too small, then, as can be seen in Fig.~\ref{fig:poly_approx_order} for $m=100$, the filter will be below $1$ for a few eigenvalues below $\lambda_{k-1}$, and above $0$ for a few eigenvalues after $\lambda_k$. If the estimated cutoff eigenvalue is a bit off, a similar situation will happen, with a shift towards lower or higher frequencies. In both cases, the value of the filter will still be $1$ up to some eigenvalue $\lambda_{j}$, then monotonically decreasing to $0$ up to some eigenvalue $\lambda_{l}$ and $0$ up to $\lambda_{N-1}$, with $\lambda_j < \lambda_k < \lambda_l$. In such a case, the filter will have non-zero coefficients in the range $[\lambda_k, \lambda_l]$ and thus, $\M$ will be contaminated by some elements of the space $\U_{[k+1, l+1]}$. However, these contributions will not appear too much in the energy of $\M$ since the coefficients of the filter for the eigenvalues bigger than $\lambda_k$ are smaller than all coefficients for the range $[\lambda_0, \lambda_{k-1}]$. Since our final approximation $\B_k$ is done using an SVD of $\M$, then $\B_k$ will be the best rank $k$ approximation of $\M$ by minimizing the energy of the residuals. Overall, as one can verify in the experiments of section~\ref{sec:experiments}, $\B_k$ will provide features that remain very good for the various applications that we develop, even with a very low polynomial order.

\subsection{Algorithms}
We propose in this section to summarize the procedure to obtain the approximation of the subspace $\Uk$ based, on one side, on the theoretical development of section~\ref{sec:theory}, and, on the other side, on the practical considerations of sections~\ref{sec:lambdak} and \ref{sec:filtering}.

\hyperref[algo:approx]{Algorithm~\ref*{algo:approx}} summarizes the steps of our method to approximate the Laplacian eigenspace $\Uk$ from data points. If a graph is not provided with the data, a $k$-NN graph can be constructed and its associated Laplacian computed beforehand. The algorithm takes a graph and a number $k$ as input and outputs a set of $k$ approximated eigenvectors of the graph Laplacian.

\begin{algorithm}
	\caption{Eigenspace Approximation}
	\label{algo:approx}
	\begin{algorithmic}[1]
		\State Generate $\R$ with $d=k$ cf. Section~\ref{sec:theory}
		\State Estimate $\lambda_k$ cf. Algorithm~\ref{algo:est_lk}
		\State Compute the approximated graph filter $g$ cf. Section~\ref{sec:filtering}
		\State Apply filtering: $\M = g(\La)\R$
		\State Compute an economic SVD: $\U \mathbf{SV} = \text{SVD}(\M)$
		\State Return the left singular vectors $\U$
	\end{algorithmic}
\end{algorithm}

\hyperref[algo:est_lk]{Algorithm~\ref*{algo:est_lk}} presents in details the strategy described in section \ref{sec:lambdak} for the accelerated estimation of $\lambda_k$. The main assumption here is that the distribution of the eigenvalues is uniform by part over the spectrum. We thus try to reach such segment of the spectrum where uniformity applies to fasten the discovery of the value of $\lambda_k$. Since some parts of the spectrum can be empty due to eigengaps for some classes of graphs, we implemented a dichotomic step to get a broad spectrum distribution estimate if the search does not progress.

\begin{algorithm}
	\caption{Estimation of $\lambda_k$}
	\label{algo:est_lk}
	\begin{algorithmic}[1]
		\Require $k, \lambda_{max}$ and $\La$
		\Ensure $\lambda_k$ (the $k^\text{th}$ eigenvalue of $\La$)
		\State Initialize: $\lambda_{lb}, c_{lb}, iter, c_{est} \leftarrow 0$
		\State $\lambda_{ub} \leftarrow \lmax, c_{ub} \leftarrow N$
		\State $\lambda_{est} \leftarrow k \frac{\lmax}{N}$
		\State Generate $\R$ with $d=k$
		
		\While{$c_{est} \neq k$ and $iter < max_{iter}$}
			\State Compute approximated graph filter $g$ with $\lambda = \lambda_{est}$
			\State $c_{est} \leftarrow \|g(\La)\R\|_F^2$
			
			\If{$c_{est} < k$}
				\State $\lambda_{lb} \leftarrow \lambda_{est}$
			\Else
				\State $\lambda_{ub} \leftarrow \lambda_{est}$
			\EndIf
			
			\If{$c_{lb} = c_{est}$ or $c_{ub} = c_{est}$}
				\State $\lambda_{est} \leftarrow \frac{\lambda_{lb} + \lambda_{ub}}{2}$
			\Else
				\If{$c_{est} < k$}
					\State $c_{lb} \leftarrow c_{est}$
				\Else
					\State $c_{ub} \leftarrow c_{est}$
				\EndIf
				\State $\lambda_{est} \leftarrow \lambda_{lb} + (k - c_{lb})\frac{\lambda_{ub} - \lambda_{lb}}{c_{ub} - c_{lb}}$
			\EndIf
		\EndWhile
		\State \Return $\lambda_{est}$
	\end{algorithmic}
\end{algorithm}

\subsection{Complexity analysis} \label{sec:complexity}
Steps 1 and 3 of Algorithm~\ref{algo:approx} are nonsignificant in the analysis of the overall complexity. We focus here on steps 2, 4 and 5 for which the number of operations is studied in details. Using fast filtering operations, applying our method consists of $k$ graph filtering operations at step 4, which is $\Landau(m|\E|k)$, with $m$ the order of the polynomial approximation of the filter. The SVD performed in step 5 has an additional cost of $\Landau(k^3)$ for a tall matrix of size $N$ by $k$ like here. Finally, step 2 takes $\Landau(m|\E|k)$ if we consider the amelioration proposed in section~\ref{sec:lambdak}. Thus, the overall complexity of our method is $\Landau(m|\E|k + k^3)$.

\paragraph{Comparison with IRLM~\cite{calvetti1994implicitly}}
As reminded above, the complexity of IRLM is $\Landau(h(|\E|k + k^2N + k^3))$ with $h$ a convergence factor. Thus, assuming $h$ and $m$ have similar orders, the IRLM needs at least $\Landau((h-1)k^2N)$ more operations than our method. In any reasonable application, we will have either $k < N$ or $k \ll N$, thus, the term $\Landau(hk^2N)$ will be larger than the term $\Landau(hk^3)$.

%Potential weak point: justify the assumption h and m have same orders

\paragraph{Comparison with CSC~\cite{tremblay2016compressive}}

Although the method presented in CSC is not directly an eigenspace estimation method, it does use the same mechanics of filtered random signals on the graph to obtain the spectral features. The number of filtering needed is $d$, which has to be larger than a threshold given by results presented in Theorems 3.2 and 3.4 of their paper. To simplify, we can say that $d = \gamma \log(\alpha k\log(k))$ where $\gamma$ and $\alpha$ are influenced by the precision of the distance preservation and the probability that the distance is preserved. Note that even with medium precision (e.g., $\Landau(10^{-1})$), the constants $\gamma$ and $\alpha$ will be large (i.e., $\Landau(10^3)$). This means that the overall complexity for the spectral features estimation will cost $\Landau(m|\E|\gamma \log(\alpha k\log(k)))$ operations. Finally, the $\Landau(\log(N))$ filterings required to estimate $\lambda_k$ have an added cost of $\Landau(m|\E|\log(N))$.

If we compare the complexity of our proposed method with the CSC we get the difference of number of operations:
\begin{align*}
\Delta &= m|\E|k + k^3 - m|\E|(d + \log(N)) \\
&= m|\E|(k - d - \log(N)) + k^3 \\
&= m|\E|(k - \gamma \log(\alpha k\log(k)) - \log(N)) + k^3
\end{align*}

For sparse graphs we can assume $|\E| = c_d N$, with $c_d$ the average node degree, which gives:
 \begin{align*}
\Delta &= m c_d N (k - \gamma \log(\alpha k\log(k)) - \log(N)) + k^3.
\end{align*}

In order to finish the comparison, we now need to make hypotheses on the relation between $k$ and $N$.

\vspace{0.2cm}
\setcounter{paragraph}{0}
\paragraph{If we assume that $k = \Landau(log(N))$, then, for $N$ large}
\begin{align*}
\Delta &= m c_d N (\log(N) - \gamma \log(\alpha k\log(k)) - \log(N)) + \log^3(N) \\
&= \log^3(N) - m c_d N \gamma \log(\alpha k\log(k)) < 0,
\end{align*}
with the last step following from the fact that $\log(\alpha k\log(k)) > 1$ and $\Landau(log^3(N)) < \Landau(N)$. This means, that for this regime, our method is cheaper than CSC, for large $N$.

%\item If we assume that $k = \Landau(\sqrt[3]{N})$, then, for $N$ large, 
%\begin{align*}
%\Delta &= m c_d N (\sqrt[3]{N} - \gamma \log(N) ) + N \\
%&= N (m c_d (\sqrt[3]{N} - \gamma \log(N) ) + 1) \\
%&> 0,
%\end{align*}
%with the last step coming from the fact that $\Landau(\sqrt[3]{N}) > \Landau(\log(N))$.

\vspace{0.2cm}
\paragraph{If we assume that $k = \Landau(\sqrt{N})$, then, for $N$ large}
\begin{align*}
\Delta &= m c_d N (\sqrt{N} - \gamma \log(\alpha N^{\frac12}\log(N^{\frac12})) - \log(N) ) + \sqrt{N^3} \\
% &= N (\sqrt{N} (m c_d + 1) - \log(\frac{\alpha^\gamma N^{\frac{\gamma+2}{2}}\log^\gamma(N)}{2})) \\
&= N (\sqrt{N}( m c_d + 1) - \gamma \log(\frac{\alpha\log(N)}{2}) - \frac{\gamma + 2}{2}\log(N)) \\
&> 0,
\end{align*}

with the last step coming from the fact that $\gamma > 1$ and $\Landau(\sqrt{N}) > \Landau(\log(N))$. This means that for this regime CSC will be cheaper than our method for large enough $N$. 

\vspace{0.2cm}

From the two cases described above we can assess that if $\Landau(1) \leq k \leq \Landau(\log(N))$ our method is cheaper and if $\Landau(\sqrt{N}) \leq k \leq \Landau(N)$ then CSC is cheaper. Note that in both cases the order of the filter $m$ was kept constant, but that both results hold for any $m$, even with $m = \Landau(N)$.

%\begin{figure*}[t!]
%\begin{center}
%
%\subfloat[]{\includegraphics[width=0.49\textwidth]{figures/complexity_analysis2_log.png}}
%\subfloat[]{\includegraphics[width=0.49\textwidth]{figures/complexity_analysis2_sqrt.png}}
%\caption{Comparison of the cost of our method (solid lines) compared with CSC (dashed lines) for a growing $N$. Both axes are in log scale. Two regimes for $k$ are displayed: $k = \log(N)$ (left) and $k = \sqrt{N}$ (right). For CSC, three distortion parameters are displayed: $\epsilon = [0.1, 0.05, 0.01]$}
%\label{fig:cost_comp}
%\end{center}
%\end{figure*}

%%%%%%%%%%%%%%%
% Experiments %
%%%%%%%%%%%%%%%

\section{Experiments} \label{sec:experiments}
In this section, we provide experiments whose objective is to show how our proposed methods behave in practice. First, we want to ensure that our proposed algorithms do fulfill their goals, i.e., that they provide accurate enough results and do so efficiently. Second, both as illustrations and practical applications, we show the performance of our eigenspace approximation method on typical clustering and visualization tasks.

The experiments were performed with the GSPBox~\cite{perraudin2014gspbox}, an open-source software. As we follow reproducible research principles, our implementations and the code to reproduce all our results is open and freely available\footnote{Available at \url{https://lts2.epfl.ch/reproducible-research/fears/}}.
Since our methods use random signals, it is expected that the results shall be slightly different in the details, but overall consistent.

\subsection{ Time performance analysis }

Since the complexity analysis in Section~\ref{sec:complexity} only covers asymptotically large $N$, it is also interesting to look at the cost of the algorithms for actual implementations and realistic values of $N$ and $k$. In addition to the eigenspace estimation with IRLM (eigs) and the $k$-dimensional spectral features of Compressive Spectral Clustering (CSC) mentioned in Section\ref{sec:complexity}, we consider the power method described in~\cite{boutsidis2015spectral} (power).

The data on which the different methods are evaluated consists of $N$ points of small intrinsic dimension which are randomly drawn. In addition, a knn graph with 10 neighbors is constructed from the data points. Each method is run with fixed parameters and the time is measured in total CPU time to completion. The results of the experiments can be seen in Fig.~\ref{fig:time_comp}.

Fig.~\ref{fig:time_comp_PM} shows the time needed in function of $k$ with $N$ fixed and for small values of $k$. The first note is that the power method does not scale well with $k$ and is exceedingly time-consuming for everything other than very small values of $k$ for which it performs well. Since it is order of magnitudes slower for the parameters used in the other experiments, it is not displayed in the remaining figures to keep readability. Fig.~\ref{fig:time_comp_Nfix} is the same as Fig.~\ref{fig:time_comp_PM} for larger values of $k$. We see, as expected in accordance with the complexity analysis, that above a threshold corresponding to $\sqrt{N}$ (i.e., 100), our method performs better than eigs and worse than CSC.

Fig.~\ref{fig:time_comp_log} shows the results for an exponentially growing $N$ and $k = \log(N)$. In this regime, our method outperforms both eigs and CSC for all values. The regime $k = \sqrt{N}$ is presented in Fig.~\ref{fig:time_comp_sqrt} where we can see that our method performs best up to $N = 10^6$. Above this value, CSC is best. Note that results above $N = 10^6$ for this regime are not shown due to memory limitations for eigs. 

Combined, those results confirm the conclusions drawn from the complexity analysis of Section~\ref{sec:complexity}. First, except for very small values of $k$, eigs is the most time-consuming method, even though it benefits from very optimized implementations. Second, for the $\log(N)$ regime, our method performs best for all values of $N$. For the $\sqrt{N}$ regime, our method is cheaper than CSC for $N < 10^6$. Above the limit $k = \sqrt{N}$, CSC is the cheapest method. As a final remark on these results, we need to point out that, contrarily to the other methods considered in this experiment, CSC does not compute an eigensubspace per se but only $k$-dimensional features allowing good pairwise distance measurements between data points. 

As a last remark on timing, we want to call attention to the fact that when filtering multiple random signals, all filtering operations are independent. Indeed, the signals are independent by definition and both the polynomial coefficients of the filter and the Laplacian are unaltered by the successive filtering operations. The filtering operations in our algorithms could thus easily benefits from a parallel implementation. 

\begin{figure*}[t!]
\begin{center}
\subfloat[]{
	\includegraphics[width=0.45\textwidth]{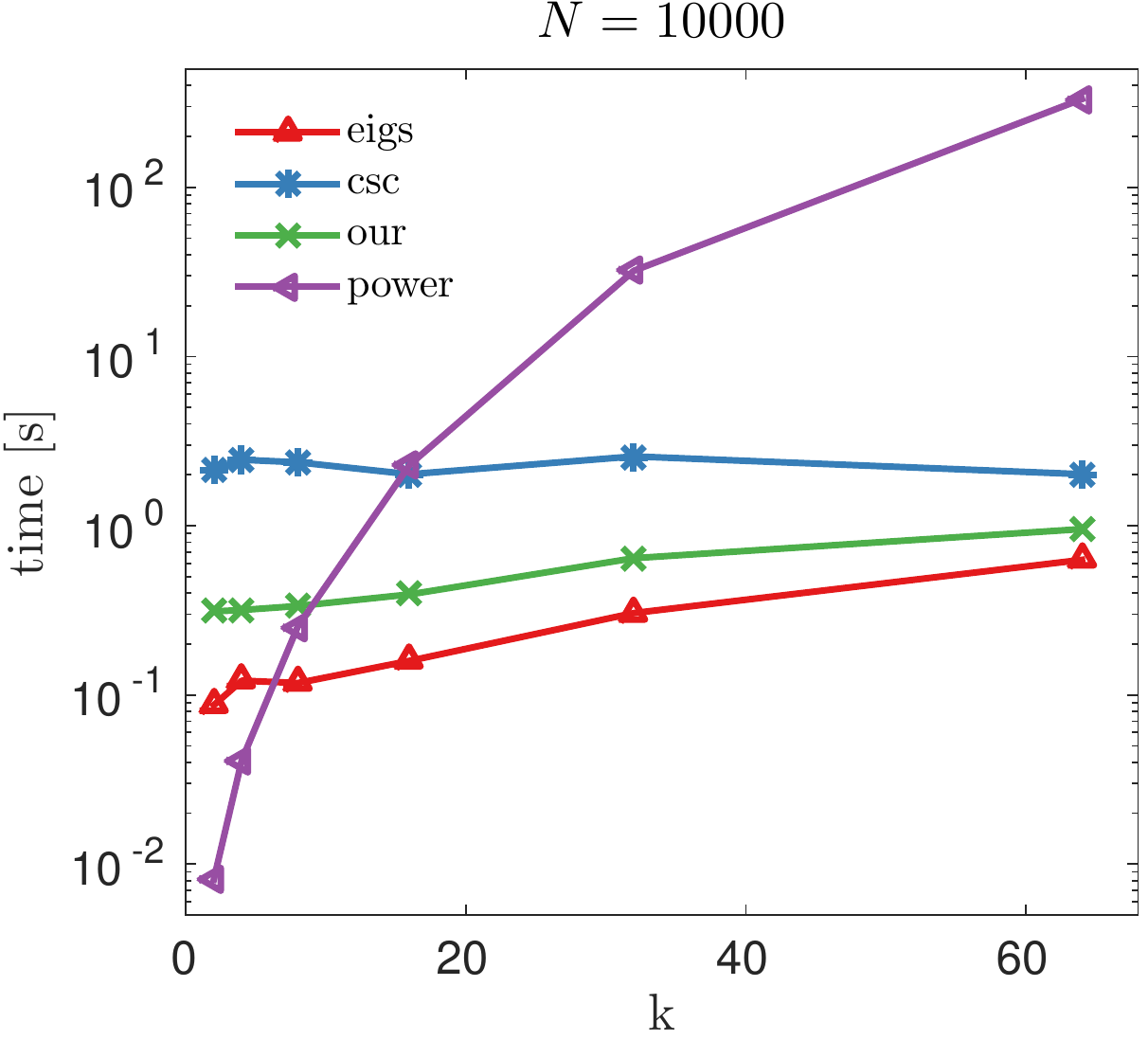}
	\label{fig:time_comp_PM}
}
\hfil
\subfloat[]{
	\includegraphics[width=0.45\textwidth]{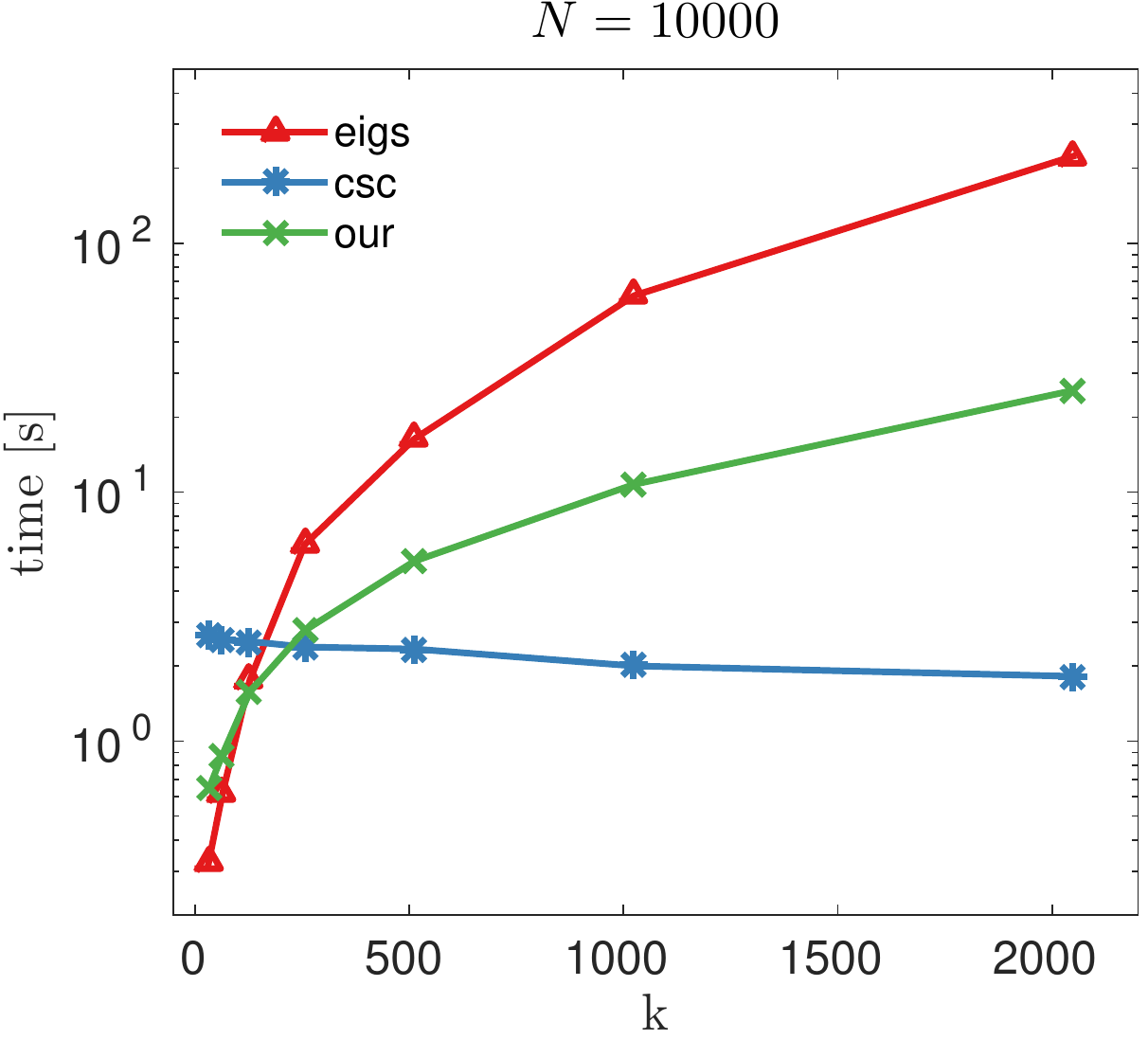}
	\label{fig:time_comp_Nfix}
}
\hfil
\subfloat[]{
	\includegraphics[width=0.45\textwidth]{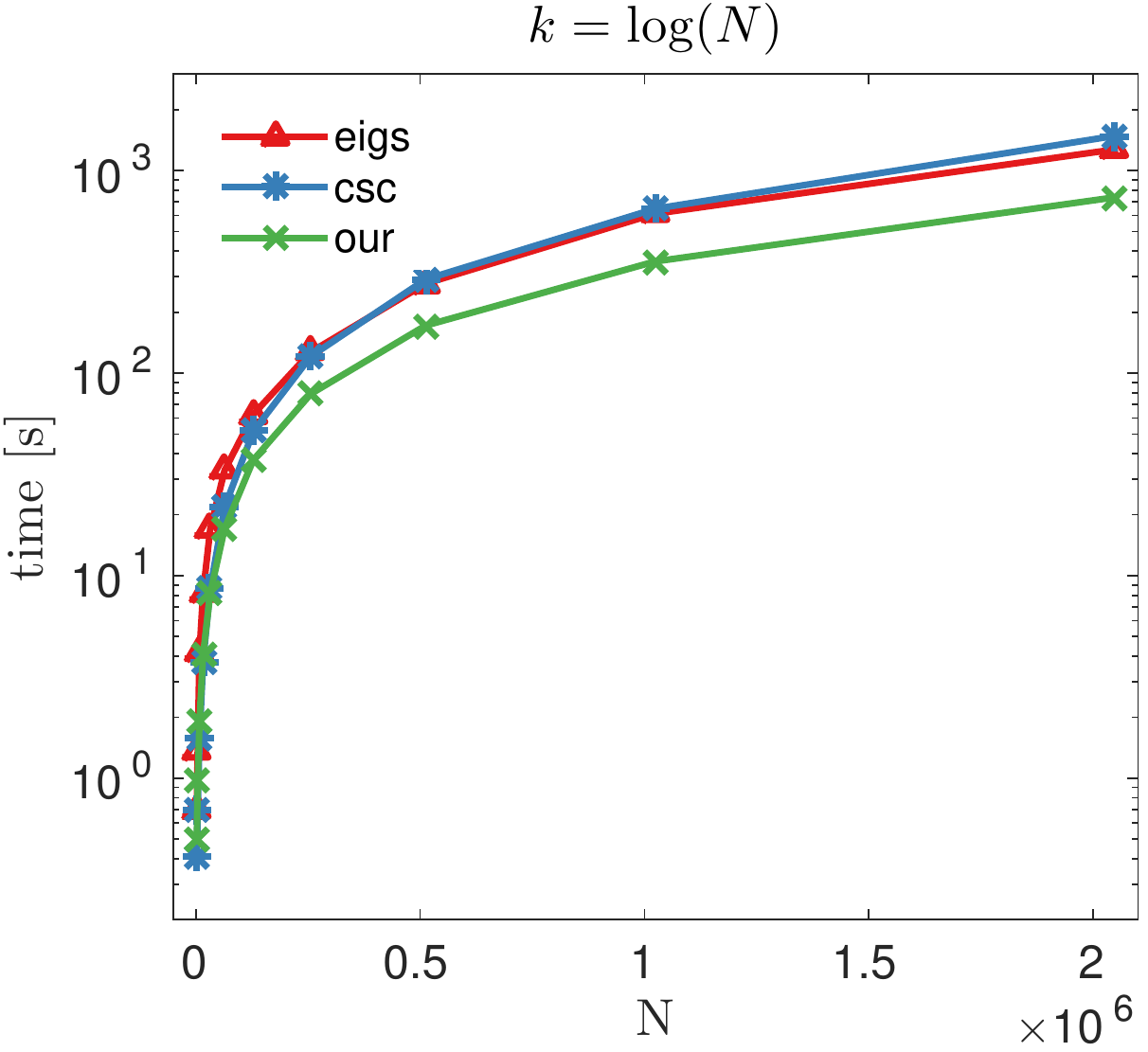}
	\label{fig:time_comp_log}
}
\hfil
\subfloat[]{
	\includegraphics[width=0.45\textwidth]{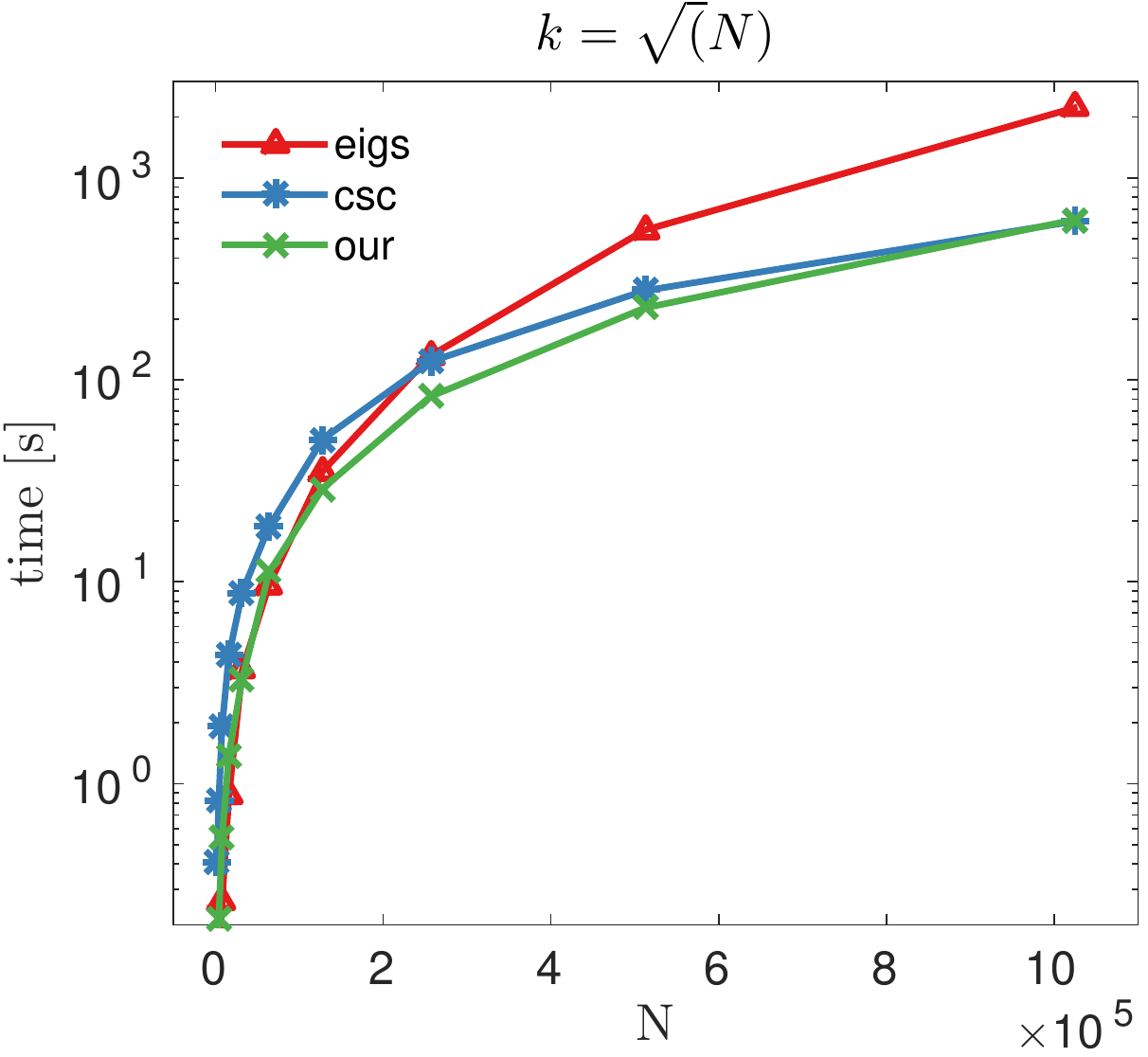}
	\label{fig:time_comp_sqrt}
}
\caption{Comparison of CPU time needed between different methods for the estimation of an eigensubspace of dimension $k$. In (a) and (b) $N$ is fixed and $k$ increases. In (c) and (d) $k$ varies in function of $N$ in two regimes ($k = \log(N)$ and $k=\sqrt{N}$ respectively). Time axis are in log-scale.}
\label{fig:time_comp}
\end{center}
\end{figure*}

\subsection{ Quality of approximation for various graphs }

In this section, we measure the accuracy of our algorithms for different classes of graphs and for different values of $k$ and $N$. In particular, we wish to evaluate two things: on one hand, the quality of approximation of the eigenspace $\U_k$ with Algorithm~\ref{algo:approx} and on the other hand the precision and efficiency of our accelerated eigencount method with Algorithm~\ref{algo:est_lk}.

The graphs chosen for this experiment are well-known classes in the field and have various spectral properties. Here is a list of all graphs with short descriptions:
\begin{itemize}
\item \textbf{Sensor network}: A graph of a synthetic sensor network, which represents randomly positioned sensors connected in a knn fashion.
\item \textbf{SBM}: Stochastic Block Model graphs model social networks or community graphs and are known to be clusterable (and thus possesses eigengaps).
\item \textbf{Swissroll}: This graph is a knn graph of the famous Swissroll manifold, a point cloud drawn from a rolled 2D surface in 3D.
\item \textbf{Bunny}: This graph is the knn graph constructed from the 3D point cloud of the Stanford bunny.
\item \textbf{Image graph}: This graph is created by connecting the pixels of an image using similarity of patches. The image of interest is the grayscale image of Barbara, a natural image often used in image processing. 
\item {\textbf{Road network}}: This graph represents the Minnesota road network (originally from the MatlabBGL library). 
\end{itemize}

In order to measure the quality of the approximated eigenspace (up to a rotation), we introduce a measure of the amount of energy which is preserved when the approximated eigenspace is projected on the real eigenspace computed with exact methods. If we note the approximated eigenspace as $\B_k$ and the exact eigenspace $\U_k$, the normalized energy kept by the projection is: 

\begin{equation}
E(\B_k, \U_k) = \frac{1}{k}\| \B_k^T \U_k \|_F^2. 
\label{eq:mean_energy}
\end{equation}

We chose to use the normalized energy to score the quality of the estimated eigenspace as it gives a number between 0 and 1 where higher values mean better approximation.

In order to compare our accelerated eigencount method with the reference dichotomy implementation of \cite{tremblay2016compressive} (abbreviated fast and standard respectively in the table), we used two measures. First, the number of iterations required until convergence, which is adequate since the workload per iteration is the same in the two algorithms. Finally, we measure how close to the actual $k$ have the algorithm converged as the mean squared deviation from $k$. This last measure is useful to state if the method was able to converge with respect to the current random matrix used for estimation, not with respect to the actual value of $\lambda_k$.
%Second, we measure the mean squared deviation from the estimated and true value of $\lambda_k$. 

The results of all measures for the various graphs described above are reported in Table~\ref{tab:proj_results}. Due to the randomness of the methods we evaluate, all experiments are averaged over 50 realizations and the standard deviation is indicated for all measures.

\begin{table}[t!]
\begin{adjustbox}{width=\textwidth}
\begin{tabular}{|cc|c|c|c|c|c|c|c|}
\cline{3-8}
\multicolumn{2}{r|}{} & Sensor network & SBM & Swiss-roll & Bunny & Image & Road network \\ \cline{3-8}
\multicolumn{2}{r|}{} & N = $10\,000$ & N = $10\,000$ & N = $10\,000$ & N = 2503 & N = $16\,384$ & N = $2\,642$ & \multicolumn{1}{r}{} \\ \cline{1-9}
\multirow{3}{*}{ME} & \multicolumn{1}{|c|}{exact} & 0.86 $\pm$ 0.01 & 1.00 $\pm$ 0.01 & 0.86 $\pm$ 0.02 & 0.99 $\pm$ 0.01 & 0.91 $\pm$ 0.01 & 0.93 $\pm$ 0.01 & \multicolumn{1}{||c|}{0.92} \\ \cline{2-9}
& \multicolumn{1}{|c|}{standard} & 0.80 $\pm$ 0.03 & 0.95 $\pm$ 0.05 & 0.79 $\pm$ 0.03 & 0.94 $\pm$ 0.05 & 0.86 $\pm$ 0.04 & 0.90 $\pm$ 0.05 & \multicolumn{1}{||c|}{0.87} \\ \cline{2-9}
& \multicolumn{1}{|c|}{fast} & 0.80 $\pm$ 0.03 & 0.96 $\pm$ 0.04 & 0.79 $\pm$ 0.03 & 0.95 $\pm$ 0.04 & 0.86 $\pm$ 0.04 & 0.90 $\pm$ 0.04 & \multicolumn{1}{||c|}{0.88} \\ \hhline{|=========|}
\multirow{2}{*}{IT} & \multicolumn{1}{|c|}{standard} & 14.62 $\pm$ 0.90 & \textbf{5.32} $\pm$ 1.58 & 4.68 $\pm$ 0.62 & 8.74 $\pm$ 1.77 & 13.06 $\pm$ 1.58 & 11.34 $\pm$ 1.22 & \multicolumn{1}{||c|}{11.29} \\ \cline{2-9}
& \multicolumn{1}{|c|}{fast} & \textbf{3.02} $\pm$ 0.71 & 9.36 $\pm$ 1.06 & \textbf{2.86} $\pm$ 0.70 & \textbf{4.48} $\pm$ 1.25 & \textbf{3.12} $\pm$ 0.75 & \textbf{3.06} $\pm$ 0.51 & \multicolumn{1}{||c|}{\textbf{4.31}} \\ \cline{1-9}
\begin{comment}
\multirow{2}{*}{LD} & \multicolumn{1}{|c|}{standard} & 0.00 $\pm$ 0.00 & 1.20 $\pm$ 0.95 &  0.00 $\pm$ 0.00 & 0.01 $\pm$ 0.01 & 0.00 $\pm$ 0.00 & \multicolumn{1}{||c|}{0.24} \\ \cline{2-8}
& \multicolumn{1}{|c|}{fast} & 0.00 $\pm$ 0.00 & 1.25 $\pm$ 0.60 & 0.00 $\pm$ 0.00 & 0.01 $\pm$ 0.01 & 0.00 $\pm$ 0.00 & \multicolumn{1}{||c|}{0.25} \\ \cline{1-8}
\end{comment}
\multirow{2}{*}{KD} & \multicolumn{1}{|c|}{standard} & 0.60 $\pm$ 0.53 & 2.46 $\pm$ 4.92 & 0.52 $\pm$ 0.58 & 0.36 $\pm$ 0.53 & 0.34 $\pm$ 0.48 & 0.36 $\pm$ 0.48 & \multicolumn{1}{||c|}{0.79} \\ \cline{2-9}
& \multicolumn{1}{|c|}{fast} & \textbf{0.00} $\pm$ 0.00 & \textbf{1.00} $\pm$ 1.01 & \textbf{0.00} $\pm$ 0.00 & \textbf{0.00} $\pm$ 0.00 & \textbf{0.00} $\pm$ 0.00 & \textbf{0.00} $\pm$ 0.00 & \multicolumn{1}{||c|}{\textbf{0.17}} \\ \cline{1-9}
\end{tabular}
\end{adjustbox}
\caption{Quality measure of our proposed methods for eigenspace estimation and $\lambda_k$ estimation$^*$. For all experiments, the following parameters were used: the order of the polynomial approximation $m = 500$, $k = 25$, $\epsilon = 10^{-1}$ for the standard eigencount method and the maximum number of iterations in fast is 10. Bold face numbers are the best score between two lines. The last column is the average over all graphs. average mean energy (ME) (as in eq.~\ref{eq:mean_energy}). The ME measure is between 0 and 1 and higher values are better. The average number of iterations is noted IT ; smaller values are better. The mean squared deviation from $k$ is noted KD ; smaller values are better. In ME, exact denotes the score computed using the true $\lambda_k$. For everything else, $\lambda_k$ is estimated either with the dichotomy method proposed in  \cite{tremblay2016compressive} (standard) or with our proposed method as in Algorithm~\ref{algo:est_lk}(fast).}
\label{tab:proj_results}
\end{table}

If we first focus on the upper part of Table~\ref{tab:proj_results} we can see that the measure of the energy (ME) using the true cutoff $\lambda_k$ shows an average above $90\%$ of precision over all graphs with a perfect score for very clusterable graphs (such as SBM) and lower values for more difficult graphs (such as Sensor network). The trend is similar using estimated values for $\lambda_k$ both with the standard and fast methods. Using the approximated cutoffs lowers the score of about $5\%$. Using the fast method leads to marginally better results. One very interesting fact regarding these results is that both the $\lambda_k$ estimation step and the eigenspace approximation contribute to the lost energy in approximately equal amounts. This tends to indicate that it is important to balance the computational effort between the two steps and not favoring one against the other.   

On the middle part of Table~\ref{tab:proj_results}, we can see the first measure reported for the eigencount evaluation. The number of iterations needed to compute $\lambda_k$ (IT) is lower with the fast method for all but the SBM graph. On average, the fast method is $2.5$ times faster than the standard method. For the SBM graph, fast is close to its maximum number of iterations meaning that the eigencount hardly converged. This result can be easily explained by the fact that the eigenvalue distribution for SBM is known to be highly non-uniform, especially for low frequencies, which is partly incompatible with the local uniformity hypothesis assumed by the fast method.

On the lower part of Table~\ref{tab:proj_results} the precision of the estimated $k$ (KD) is reported. Both the fast and standard method converge most of the time, with a better overall convergence of the former which converges exactly to the true value except for SBM. This could be expected from the high number of iterations needed for this specific graph. 

From those results, we can see that the quality of the estimated subspace computed using our proposed method is decent, while not perfect. The imprecision coming both from the approximation in the filter design and cutoff eigenvalue estimation. Our scheme for accelerated $\lambda_k$ estimation is faster than the reference method and provide very good results.

\subsection{Clustering}
This experiment proves the capability of our filtered signals (us) to produce an assignment for the data points. We will compare the results obtained by our method to Spectral Clustering (SC) \cite{shi2000normalized} and Compressive Spectral Clustering (CSC) \cite{tremblay2016compressive}. We will also see that the compressive step of the latter can be used with $k$ filtered signals instead of $d$.

\paragraph{Spectral clustering}
Spectral clustering is a very famous method that follows directly from the relaxation of NCut for $k$ classes. It states that the $k$ eigenvectors associated with the smallest eigenvalues of $\La$ are the optimal solution of the optimization problem of NCut. Thus, by computing the eigensubspace $\Uk$ one easily gets a very good assignment for the data partitioning problem since the $k$-means solution over the rows of the matrix $\Uk$ gives a standard discrete partition of the data points. However, computing spectral clustering on large graphs is not to be considered due to the runtime complexity of the method ($\Landau(N^3)$ for exact methods, $\Landau(k^2N)$ with IRLM).

\paragraph{Compressive spectral clustering}
In this work, the authors replaced the features formed by the eigenvectors with filterings of random signals on the graph $\G$. They propose a minimal number of signals to filter in order to preserve the distances between any two points in the data set. There only remains to apply $k$-means on the filtered signal to obtain an assignment identical to spectral clustering. Their second contribution is to show that $k$-means can be compressed, in the sense that only a subset of the nodes needs to be assigned with this costly method. The remaining labels can be inferred by solving an optimization problem using graph regularization.

\vspace{0.2cm}
\subsubsection{Synthetic case: Stochastic Block Model}
For this experiment, we use a Stochastic Block Model (SBM) with $N = 5\,000$ nodes and $k = 20$ clusters. We set the average degree of the nodes to $s = 16$ and the nodes are associated at random with a particular class (the ground truth for the assignment). Then, an edge between two nodes exists with probability $p$ if the two nodes belong to the same class (intra-cluster probability) and with probability $q < p$ if they belong to different clusters (inter-cluster probability). We generate several graphs with different ratios $\varepsilon = \frac{q}{p}$ (the larger $\varepsilon$, the harder the community detection) to evaluate our clustering capabilities in the task.

The evaluation of the presented methods is performed using the adjusted Rand similarity index~\cite{hubert1985comparing} between the SBM ground truth and the resulting assignments. All results presented here are averaged over 50 realizations in each setup. By looking at Fig.~\ref{fig:SBM_clust} we can first observe that our method is the one that approximates the best the results of SC. It is not necessarily the method achieving the best rand index as $\varepsilon$ increases but the ground truth is set before the edges are created. Thus, for relatively large values of $\varepsilon$, it might not make sense to keep this assignment for clustering purposes. In our view, spectral clustering is the target to fit at best. Moreover, notice that the order of the polynomial approximations alters the result of the clustering in both our method and CSC. Finally, Cus represents the result of our features assigned with the compressive step of CSC instead of the full $k$-means. We see that $k$-means is more faithful to spectral clustering than the regularized label diffusion on the graph.

\begin{figure}[t]
\begin{center}
\includegraphics[width=0.75\columnwidth]{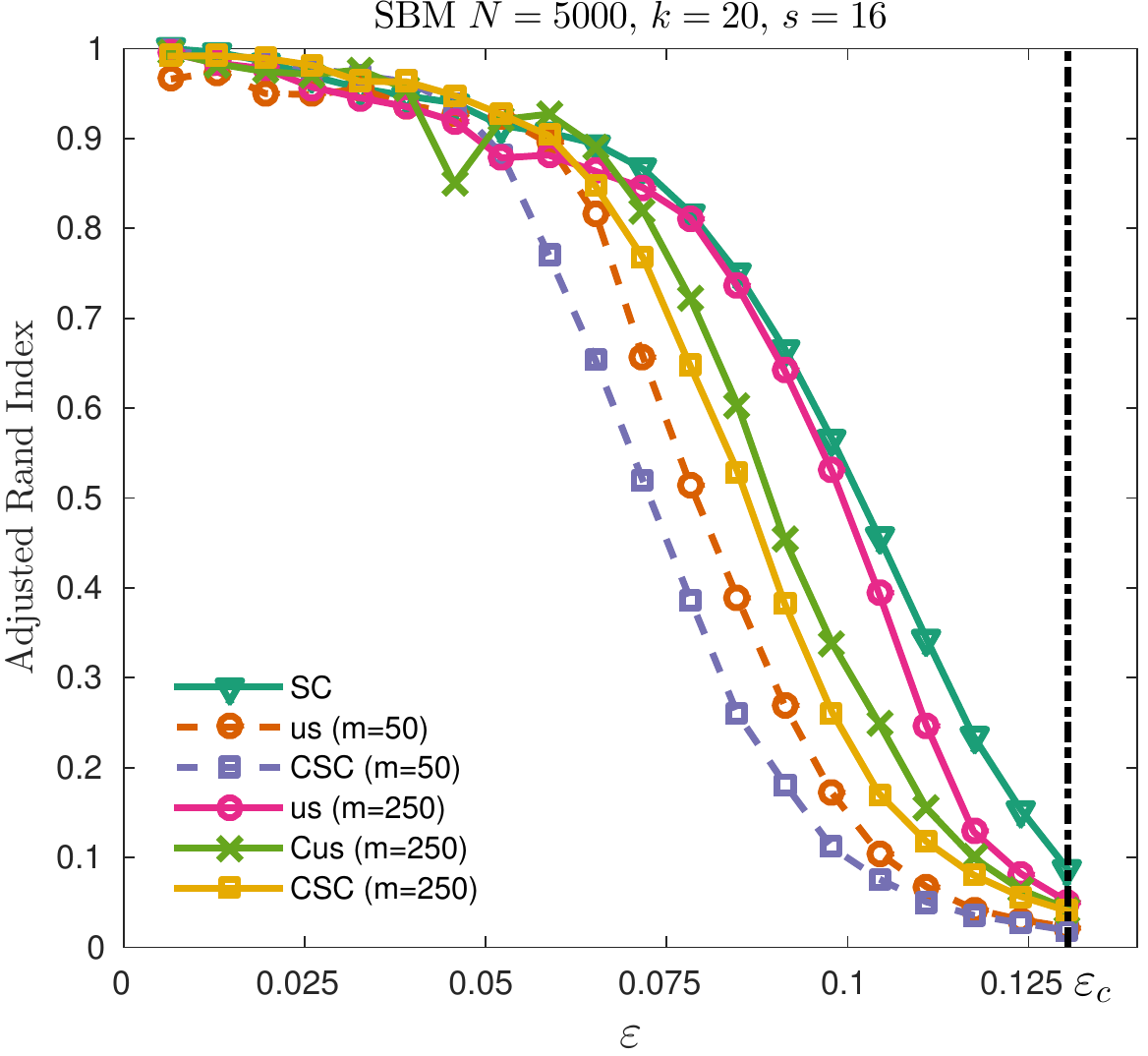}
\caption{Study of the clusterability of Stochastic Block Models for various values of $\varepsilon$, representing how well the graph can be split into clusters. Our method is the best to approximate the result of spectral clustering.}
\label{fig:SBM_clust}
\end{center}
\end{figure}

\vspace{0.2cm}
\subsubsection{Real-world example: Amazon co-purchasing network}
In addition to the synthetics SBM graphs, we want to go further and show that this also works well for real-world data sets. To this end, we consider the problem of clustering the Amazon co-purchasing network~\cite{yang2015defining} that has also been evaluated for the study of CSC. The graph is composed of 334\,863 nodes and 925\,872 edges\footnote{Available at \url{http://snap.stanford.edu/data/com-Amazon.html}}. No clear ground truth can be used to compare against since the given information is the belonging of products to categories with overlaps. We decided to reproduce the experiment published in~\cite{tremblay2016compressive}, adding our method to the benchmark. We measured the resulting assignments with two measures: the modularity score~\cite{newman2004finding}, used to determine whether a given partition is separating the network efficiently, and the adjusted Rand similarity index compared to the result of SC, used to identify the resemblance of the two assignments.

In Table~\ref{tab:timing_amazon} we first show the performance of the different algorithms with 3 different numbers of clusters: 250, 500 and 1\,000. We split the timing into two parts, one for the feature extraction process and the other for the assignment based on these features. We see that consistently the features extracted using random signal filtering are faster to compute than those requiring partial eigendecomposition. We also notice that until $k=500$, $k$-means is an efficient method for the assignment of the points to the clusters, it is even 5 times faster than the compressive assignment for $k=250$ in our experiment. However, when $k$ becomes larger, using the compressive method of CSC (also applied in Cus) is helping greatly to reduce the overall time of the computation, earning a factor 2 speedup between us and Cus.

Next, we consider the efficiency of the clustering reported in Table~\ref{tab:clust_amazon}, where two important observations stand out. On one hand, the best modularity is achieved using CSC and we see that our method, with the use of the compressive step, tends to similar results with increasing $k$. On the other hand, the adjusted Rand similarity index clearly shows that our method is assigning the nodes very similarly to SC. This is an expected behavior since the goal of our method is to reconstruct the set of the $k$ first eigenvectors used as features in SC.

\begin{table}[t!]
\centering
\begin{tabular}{|c|c|c|c|}
\cline{2-4}
\multicolumn{1}{r|}{} & $k = 250$ & $k = 500$ & $k = 1\,000$ \\ \hline
SC & 14.37min + 2.13h & 25.09min + 14.96h & 55.63min + 106.87h\\ \hline
us & 0.12min + 2.55h & 0.19min + 22.75h & 0.52min + 104.82h\\ \hline
Cus & 0.12min + 11.36h & 0.19min + 17.22h & 0.52min + 58.46h \\ \hline
CSC & 2.34min + 9.74h & 3.73min + 21.07h & 2.61min + 35.47h \\ \hline
\end{tabular}
\caption{Timing of clustering for Amazon data set. All values represent one experiment and the order of the polynomial approximation is $m = 500$. Each experiment is split into two steps: the computation of the features (in minutes), and the assignment from the features to a cluster (in hours).}
\label{tab:timing_amazon}
\end{table}

\begin{table}[t!]
\centering
\begin{tabular}{|c|c|c|c|c|c|c|c|}
\cline{2-8}
\multicolumn{1}{r|}{} & SC & \multicolumn{2}{c|}{us} & \multicolumn{2}{c|}{Cus} & \multicolumn{2}{c|}{CSC}\\ \cline{2-8}
\multicolumn{1}{r|}{} & mod\tnote{1} & mod & rand\tnote{2} & mod & rand & mod & rand\\\hline
$k=250$ & 0.344 & 0.387 & 0.884 & 0.588 & 0.711 & 0.764 & 0.509\\ \hline
$k=500$ & 0.507 & 0.605 & 0.818 & 0.759 & 0.677 & 0.818 & 0.586\\ \hline
$k=1\,000$ & 0.663 & 0.638 & 0.851 & 0.815 & 0.780 & 0.798 & 0.749\\ \hline
\end{tabular}
\caption{Evaluation of clustering for Amazon data set. All values are representing one experiment and the order of the polynomial approximation is $m = 500$. The modularity score (\cite{newman2004finding}) is noted mod and the adjusted Rand similarity index (\cite{hubert1985comparing}) rand.} 
\label{tab:clust_amazon}
\end{table}

\subsection{Visualization}

In this last experiment, we show how our method can be used in the context of visualizing high-dimensional data, since eigenspaces are commonly used for dimensionality reduction in this context. We wish to see how our proposed method behaves first in a very simple synthetic example and second for real-world data sets of larger size. For this task we compare the following visualization algorithms:

\paragraph{Laplacian eigenmaps}

Belkin and Niyogi~\cite{belkin2001laplacian} proposed to solve the generalized eigenvalues problem $\La \mathbf{y} = \lambda D \mathbf{y}$ where $y$ is called the Laplacian eigenmaps. This method is interesting to validate the fact that our method finds a good approximation of $\U_k$ because it finds the eigenspace of the random walk Laplacian. Indeed, if we define the random walk Laplacian as $P = D^{-1}\La$ then the equation above can be rewritten as $P\mathbf{y} = \lambda \mathbf{y}$. Thus, Laplacian eigenmaps aims at finding the eigenspace of $P$ and use it as an embedding for visualization. We implemented the method in Matlab with the eigs eigensolver which uses the IRLM algorithm. 

\paragraph{t-SNE~\cite{maaten2008visualizing}}

a famous state-of-the-art technique for visualization which enhanced the Stochastic Neighbor Embedding method~\cite{hinton2002stochastic}. The use of a heavy-tail distribution for the embedded points probabilistic model allows avoiding the crowding effect and at the same time gives rise to an easier optimization problem. The original implementation having an $\Landau(N^2)$ complexity, the Barnes-Hut accelerated version is often used for large data sets since it has a $\Landau(N\log(N))$ complexity. We used the C++ implementation of the Barnes-Hut t-SNE for our experiments\footnote{Available at \url{https://github.com/ninjin/barnes-hut-sne}}.

\paragraph{LargeVis~\cite{tang2016visualizing}}
a recent technique based on graph visualization which aims at solving the scalability problems of state-of-the-art methods such as t-SNE. Its first contribution is to accelerate the graph construction step by using an approximated k-NN graph construction method. Second, it formulates the embedding problem as a probabilistic model which keeps similar vertices close to each other and dissimilar vertices apart. Inspired by negative sampling techniques they propose to optimize the probabilistic model using independent stochastic gradient descent steps. The C++ implementation of the algorithm was used for the experiments\footnote{Available at \url{https://github.com/lferry007/LargeVis}}. 

\vspace{0.2cm}
\subsubsection{Toy example: the Swissroll}

In this first small experiment, we wish to assess the validity of using our proposed method of eigenspace estimation for visualization on a simple toy example. We will compare the results obtained by our method only with Laplacian Eigenmaps as we would like to verify that we get similar results. 

For this experiment, we use a classical Swissroll graph for which we compute a 2 dimensional embedding. The Swissroll is computed by sampling its continuous manifold in the following way: given a set of randomly drawn angles $\theta$ in $[a\pi, b\pi]$ the coordinates are set as $x = \theta  \cos(\theta)$,  $y$ drawn uniformly in $[0, 1]$ and $z = \theta \sin(\theta)$. A knn graph with 10 neighbors is constructed from the data points. For this experiment, the normalized Laplacian was used for all methods.

The resulting embeddings are shown in Fig.~\ref{fig:viz_small}. The colormap is a linear function of $\theta$. The first thing to notice is that all embeddings are very smooth with respect to $\theta$. The second interesting fact is that $\B_k$ indeed seems to be a good approximation of the Laplacian eigenmaps up to a rotation as they have very similar shapes. This tends to validate that the method indeed provides a good approximation of $\U_k$. In addition, in this specific example, while embedding with $M$ gives a smooth result, the normalization step provided by the SVD is necessary to get a good enough visualization. This observation makes sense as for visualization very few random signals are used to get $M$, which, as discussed in Section~\ref{sec:mapprox}, is not sufficient to have an excpectation effect smoothing the variance on the eigenvalues. This scaling is normalized by the final SVD step, which is not costly for visualization tasks since $k$ is very small.  %Finally, all mappings are not perfectly square with respect to the colormap, which can be due to the fact that none of the embedding methods were tuned for this particular example. 
 
\begin{figure*}[ht!]
\begin{center}
\subfloat[Swissroll]{\includegraphics[width=0.4\textwidth]{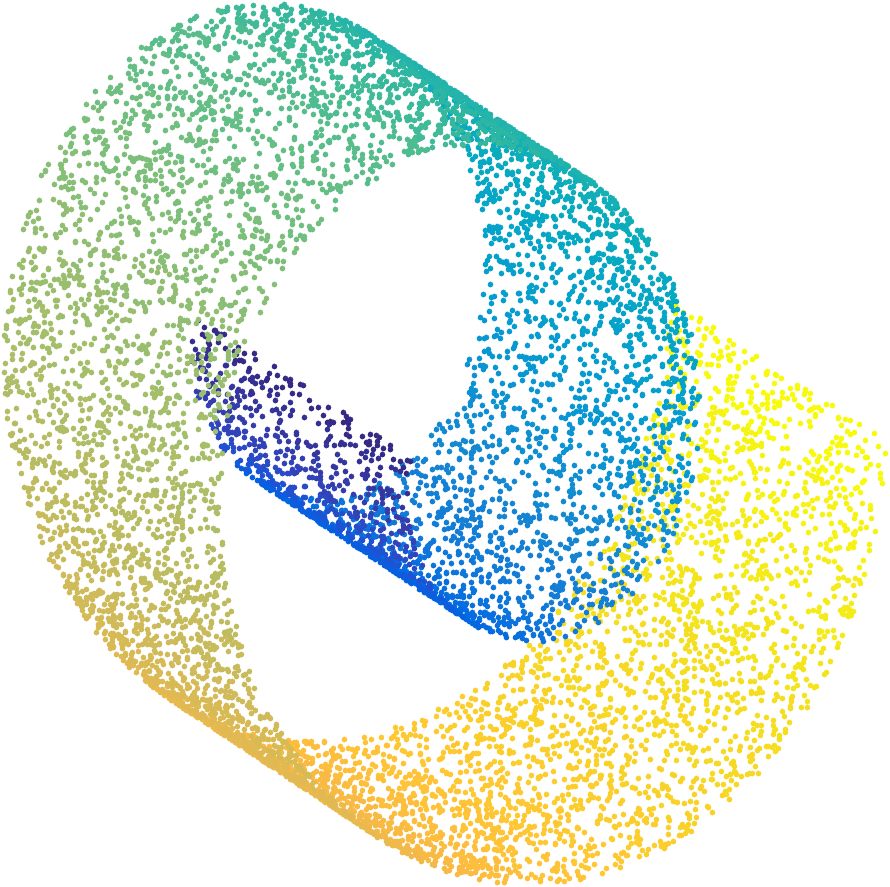}}
\hfil
\subfloat[Laplacian eigenmaps]{\includegraphics[width=0.4\textwidth]{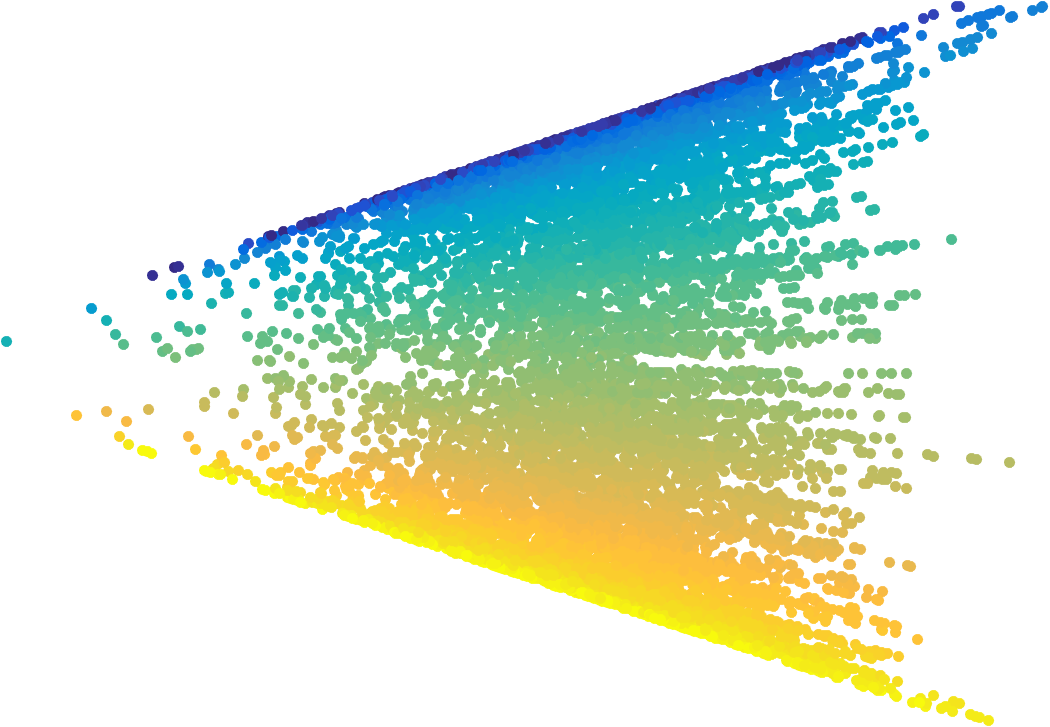}}
\hfil
\subfloat[Our method with $\M$]{\includegraphics[width=0.4\textwidth]{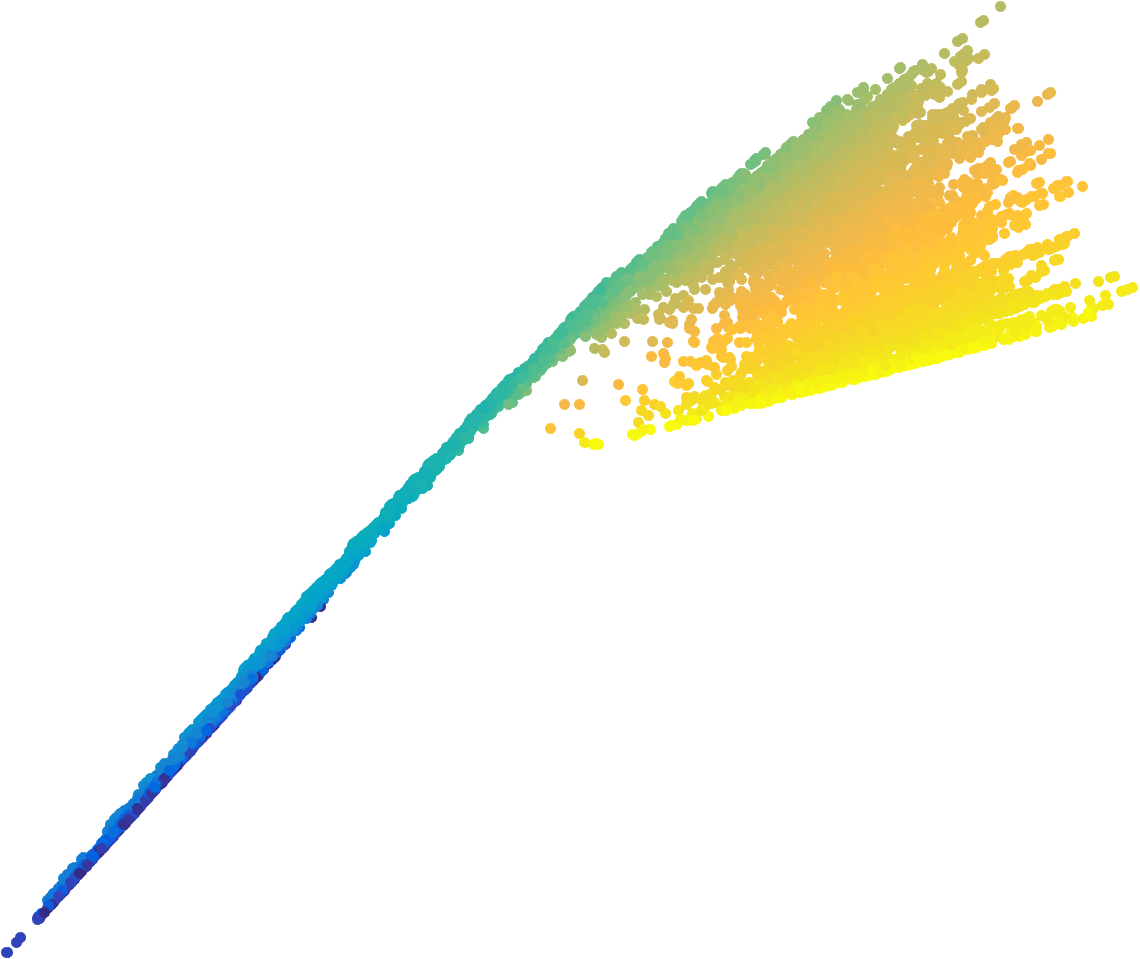}}
\hfil
\subfloat[Our method with $\B_k$]{\includegraphics[width=0.4\textwidth]{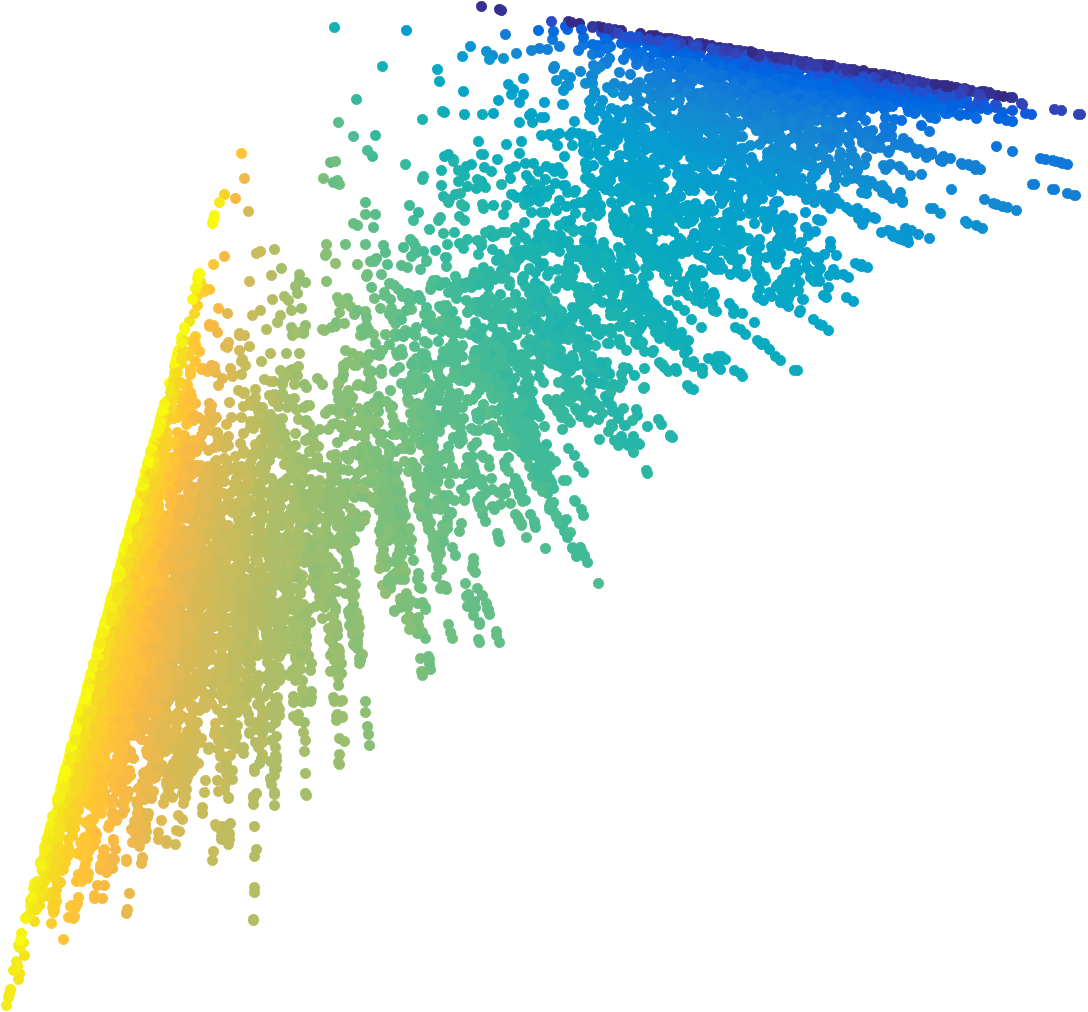}}
\caption{ The Swissroll point cloud (a) with $10\,000$ nodes and its 2D embeddings using Laplacian eigenmaps (b), our proposed fast eigenspace estimation method prior to the SVD step (c), and after the SVD step (d). }
\label{fig:viz_small}
\end{center}
\end{figure*}

\vspace{0.2cm}
\subsubsection{Real-world data sets}
In this second experiment, we will consider large scale real-world examples and compare our method with existing approaches. We will use the two following data sets:

\paragraph{MNIST} a well known data set of handwritten digit images, from which we take all $70\,000$ data points \footnote{Available at \url{http://yann.lecun.com/exdb/mnist/}}. %All points are labelled with the digit it represent. 

\paragraph{LiveJournal} a data set from the LiveJournal social network. The graph used is the largest connected componant of the complete graph wich has $3\,997\,962$ nodes\footnote{Available at \url{http://snap.stanford.edu/data/com-LiveJournal.html}}. 

In Fig.~\ref{fig:viz_large} we can see the visualizations of the MNIST data set, where  the colormap comes from the labels. The first observation is that both Laplacian eigenmaps and our proposed method yield similar results. Both do not achieve a very good separation of the classes and suffer from a concentration around the origin (i.e., the crowding problem). Our method seems to do a slightly better job at separating the classes in the middle than Laplacian eigenmaps. The embeddings provided by both t-SNE and LargeVis are of much greater quality with respect to class separation even if they leave outliers. Also, both methods find 11 clusters instead of 10 as they split one class into two clusters. 

\begin{figure*}[t!]
\begin{center}

\subfloat[Laplacian eigenmaps]{\includegraphics[width=0.4\textwidth]{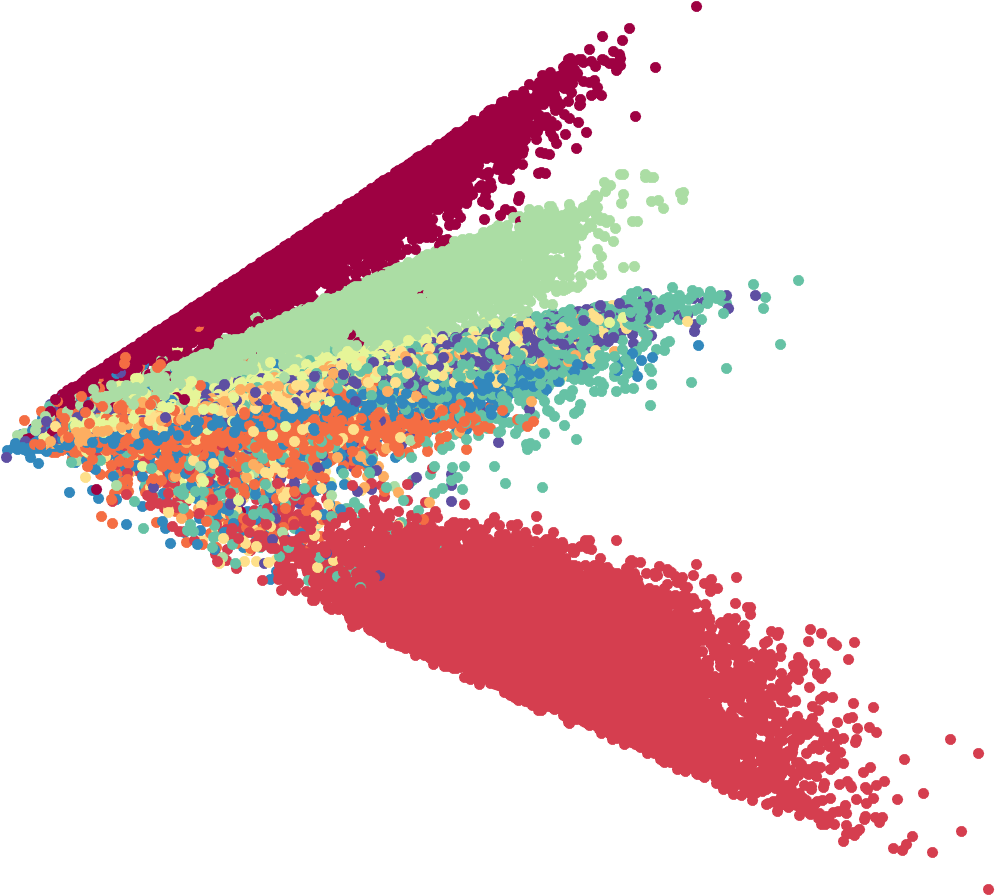}}
\hspace{1em}
\subfloat[Our method]{\includegraphics[width=0.4\textwidth]{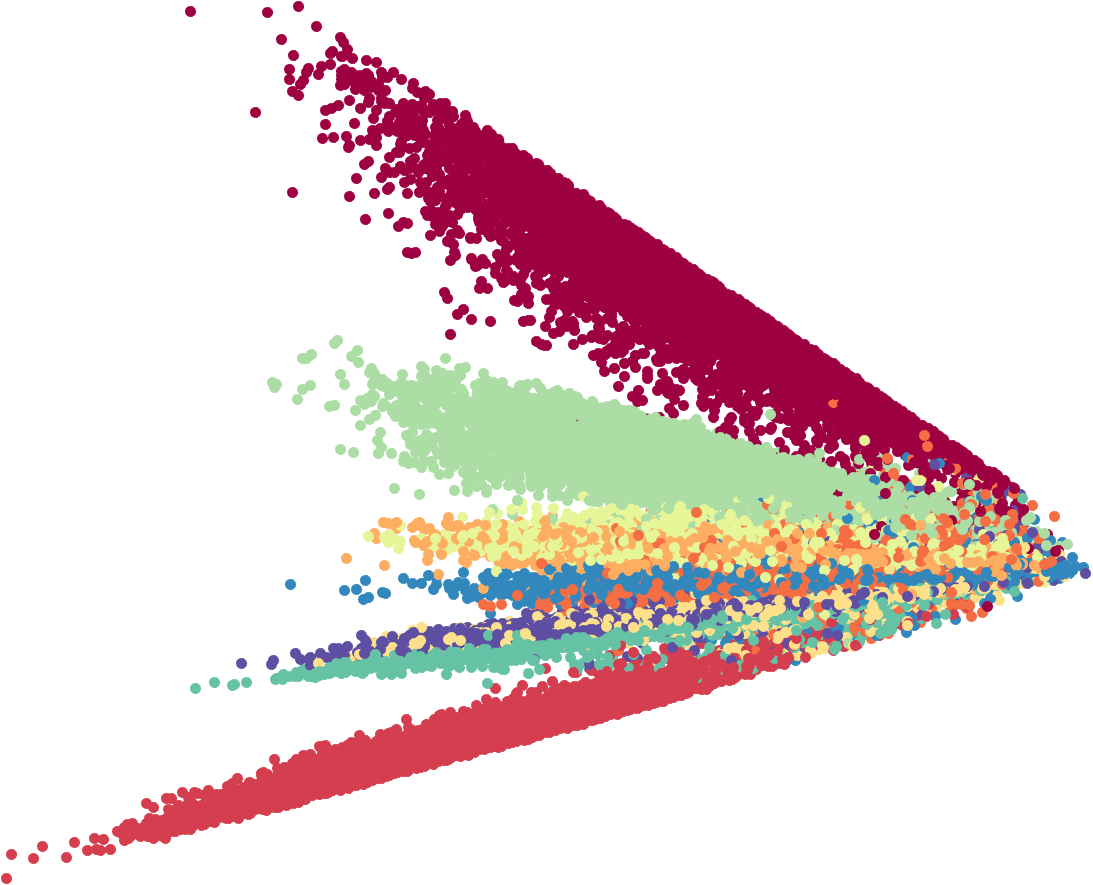}}
\hspace{1em}
\subfloat[t-SNE]
{\includegraphics[width=0.4\textwidth]{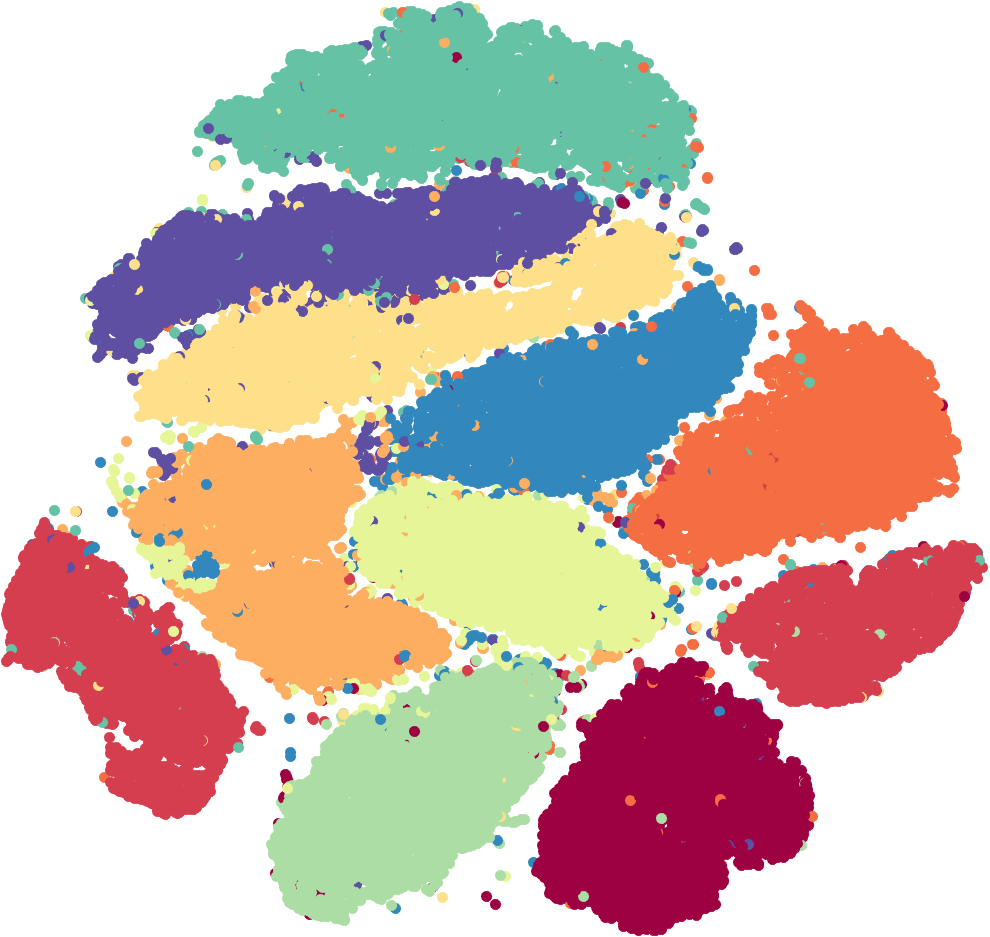}}
\hspace{1em}
\subfloat[LargeVis]{\includegraphics[width=0.4\textwidth]{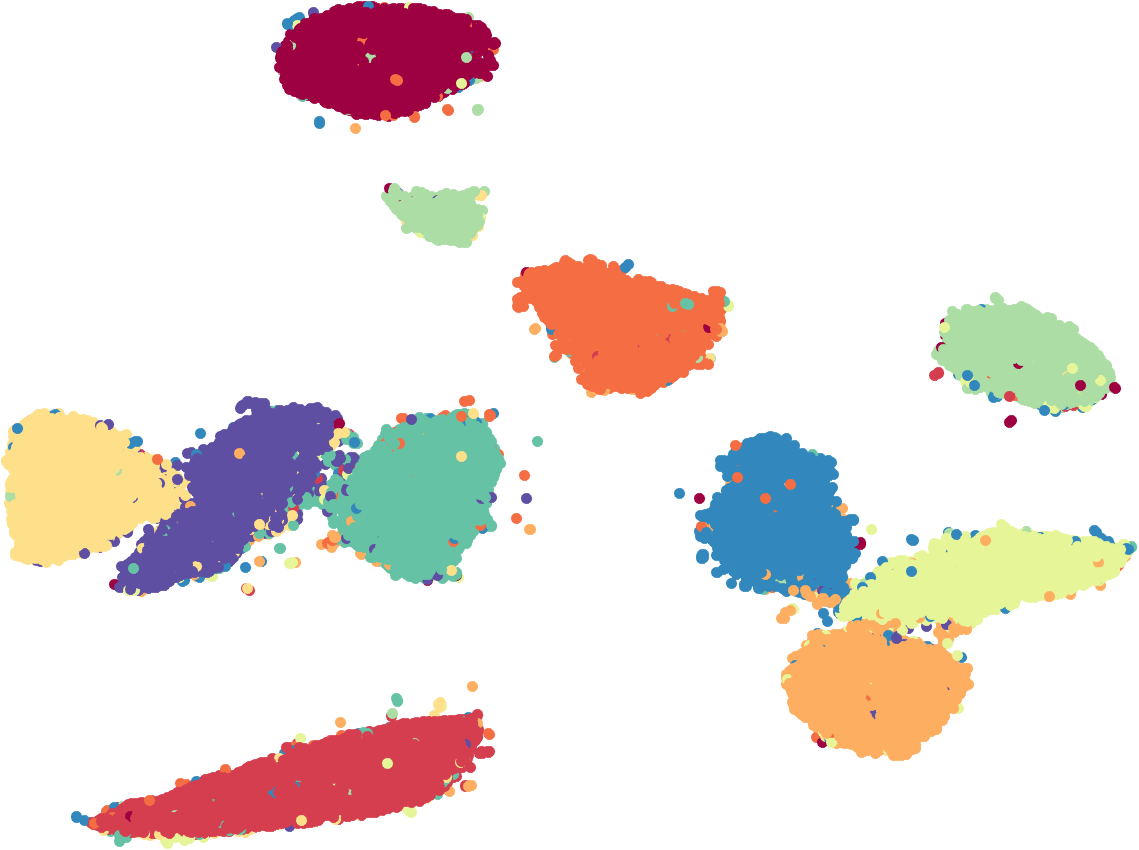}}
\caption{Visualizations of MNIST using (a) Laplacian eigenmaps, (b) our method, (c) t-SNE (with Barnes-Hut implementation) and (d) LargeVis. The colors correspond to the different categories (i.e., numbers from 0 to 9).}
\label{fig:viz_large}
\end{center}
\end{figure*}

In Table~\ref{fig:viz_large_timing} we report the time needed to compute the embeddings using the methods above on the two data sets. The timing of t-SNE and LargeVis on LiveJournal is based on adjusted reported results from~\cite{tang2016visualizing}. On MNIST, our method has the lowest CPU time, closely followed by Laplacian eigenmaps. Our method is one order of magnitude faster than t-SNE, which is twice slower than LargeVis. On LiveJournal, our method is still the fastest and one order of magnitude faster than t-SNE. LargeVis, while being slower than our method performs rather well. Laplacian eigenmaps exceeded the available memory and did not complete. 

From these results we can say that our method is valid for visualization but cannot achieve a quality close to state-of-the-art methods such as t-SNE or LargeVis. However, it has the advantage to be fast and scales well even using a non-optimized mono-thread implementation.  

\begin{table}[t!]
\centering
\begin{tabular}{|c|c|c|c|c|}
\cline{2-5}
\multicolumn{1}{r|}{Time [h]} & Eigenmaps & t-SNE & LargeVis\tnote{1} & Our \\ \cline{1-5}
MNIST & 0.06 & 0.46 & 0.26 & \textbf{0.04}     \\ \cline{1-5}
LiveJournal & -\tnote{2} & 78.79 & 10.37 & \textbf{5.80} \\ \cline{1-5}
\end{tabular}
\caption{2D Embedding computation time. The default implementation of LargeVis uses parallelism. The value for Eigenmaps on LiveJournal is not reported because it exceeded the maximum memory available (128 GB). }
\label{fig:viz_large_timing}
\end{table}
% NOTE: LargeVis on MNIST: 0.23 hours, ~828 s => rounded to this perf factor: 1168s
% NOTE: LargeVis on LiveJournal: t-SNE 

%%%%%%%%%%%%%%
% Conclusion %
%%%%%%%%%%%%%%

\section{Conclusion} \label{sec:conclusion}
In this contribution, we have presented a theoretical way to recover exactly the set of $k$ smallest eigenvectors of a graph Laplacian. We have shown an accelerated algorithm for the approximation of the eigenspace of the Laplacian $\La$ solely based on Gaussian random signals filtering. We proved the bound on the number of signals to be as tight as ever possible. In addition, we proposed an accelerated eigenvalue estimation algorithm based on eigencount techniques. We presented different applications and compared the efficiency against the state of the art, showing the ability for our method to scale with very large $N$.

This is an interesting result for the field of graph signal processing and many further questions arise in this context. Among them, the design of the filter could be reconsidered. Could we gain even more efficiency by using a naturally polynomial function for the filter instead of the approximation of an ideal low-pass filter? We suggest using exponentially decreasing kernels, which are low-pass and infinitely differentiable and will assign to the eigenvalues an energy proportional to its position in the spectrum. One could wonder whether such design could allow stopping the computation before the SVD step.

\appendix

\appendix[Properties of projected Gaussians]
We stated in section \ref{sec:theory} that a Gaussian random matrix projected over a basis keeps its Gaussian properties. We will demonstrate the different properties in this appendix.

Let $\U \in \Rbb^{N\times N}$ describe a basis of $N$ orthonormal vectors and $\R \in \Rbb^{N\times d}$ be a Gaussian random matrix with \iid entries~$\nrm{0}{\sigma^2}$.

Mathematically,
\begin{equation}
	\forall i, j: \left(\U\,\R \right)_{i,j} = \scp{u_{i-1}}{r_j} = \sum_{\ell=1}^N{u_{i-1}(\ell) r_{\ell, j}},
\end{equation}
is a linear transformation of the elements of $\R$. Thus, there are Gaussians. Moreover, we already knew that the size of the product is a $N \times d$ matrix. Next, we will evaluate the two first moments of all those entries.
\begin{align}
	\Esp\left[\sum_{\ell=1}^N{u_{i-1}(\ell) r_{\ell, j}}\right] & = \sum_{\ell=1}^N{u_{i-1}(\ell) \Esp[r_{\ell, j}]} = 0\\
	\vr\left(\sum_{\ell=1}^N{u_{i-1}(\ell) r_{\ell, j}}\right) & = \sum_{\ell=1}^N{u_{i-1}^2(\ell) \vr(r_{\ell, j})} \nonumber \\
	& = \sigma^2 \sum_{\ell=1}^N{u_{i-1}^2(\ell)} = \sigma^2
\end{align}

This shows that all entries of $\U\,\R$ are identically distributed. Then we can compute the covariance between any two entries $\bigl(\left(\U\,\R\right)_{i, j}$ and $\left(\U\,\R\right)_{n, m}\bigr)$ to ensure independance: 
\begin{align*}
	\cv\left(\U\,\R\right) & = \Esp\left[\sum_{\ell=1}^N{u_{i-1}(\ell) r_{\ell, j}} \sum_{\ell'=1}^N{u_{n-1}(\ell') r_{\ell', m}}\right]\\
		& = \sum_{\ell=1}^N{\sum_{\ell'=1}^N{u_{i-1}(\ell) u_{n-1}(\ell') \Esp[r_{\ell, j} r_{\ell', m}]}}\\
		& = \mathds{1}_{\{m=j\}} \sum_{\ell=1}^N{u_{i-1}(\ell) u_{n-1}(\ell) \Esp[r_{\ell, m}^2]}\\
		& = \sigma^2 \mathds{1}_{\{m=j\}} \langle u_{i-1}, u_{n-1} \rangle\\
		& = \sigma^2 \mathds{1}_{\{m=j\}} \mathds{1}_{\{n=i\}},
\stepcounter{equation}\tag{\theequation}
\end{align*}
which shows that any two entries in $\U\,\R$ are independant. Combining the last two shows that the entries of $\U\,\R$ are \iid Gaussian random samples with pdf $\nrm{0}{\sigma^2}$ just like $\R$.

\section*{Acknowledgment}
We would like to thank Dr. O. L\'ev\^eque for our insightful discussions on eigenvalues distributions. We are also thankful to Dr. A. Loukas and Prof. P. Vandergheynst for their precious advices regarding this work.

\bibliographystyle{ieeetr}
\bibliography{biblio}

\end{document}